\documentclass{article}
\usepackage[backend=bibtex,sorting=none,citestyle=numeric-comp,citestyle=ieee,doi=false,isbn=false,url=false]{biblatex}
\usepackage[margin=0.7in]{geometry}
\usepackage{amsmath,amsthm,amssymb,amsfonts,bm,mathtools,braket,units,array,comment,cases}
\usepackage{ulem}
\usepackage{color}
\usepackage{enumitem}
\usepackage[hidelinks]{hyperref}
\usepackage[displaymath,mathlines]{lineno}
\usepackage[space]{grffile}
\usepackage{mathpazo}
\usepackage{microtype}
\usepackage{authblk}
\usepackage{booktabs}
\usepackage{tikz}
\usetikzlibrary{matrix}
\usetikzlibrary{shapes.geometric}
\usetikzlibrary{arrows.meta}
\usetikzlibrary{positioning,shapes,arrows}

\newcommand{\analysispath}{.}

\usepackage{etoolbox} 

\newcommand*\linenomathpatch[1]{%
	\expandafter\pretocmd\csname #1\endcsname {\linenomath}{}{}%
	\expandafter\pretocmd\csname #1*\endcsname{\linenomath}{}{}%
	\expandafter\apptocmd\csname end#1\endcsname {\endlinenomath}{}{}%
	\expandafter\apptocmd\csname end#1*\endcsname{\endlinenomath}{}{}%
}
\newcommand*\linenomathpatchAMS[1]{%
	\expandafter\pretocmd\csname #1\endcsname {\linenomathAMS}{}{}%
	\expandafter\pretocmd\csname #1*\endcsname{\linenomathAMS}{}{}%
	\expandafter\apptocmd\csname end#1\endcsname {\endlinenomath}{}{}%
	\expandafter\apptocmd\csname end#1*\endcsname{\endlinenomath}{}{}%
}

\expandafter\ifx\linenomath\linenomathWithnumbers
\let\linenomathAMS\linenomathWithnumbers
\patchcmd\linenomathAMS{\advance\postdisplaypenalty\linenopenalty}{}{}{}
\else
\let\linenomathAMS\linenomathNonumbers
\fi

\linenomathpatch{equation} 
\linenomathpatchAMS{gather}
\linenomathpatchAMS{multline}
\linenomathpatchAMS{align}
\linenomathpatchAMS{alignat}
\linenomathpatchAMS{flalign}

\newtheorem{theorem}{Theorem}
\newtheorem{proposition}[theorem]{Proposition}

\newtheorem{corollary}[theorem]{Corollary}
\newtheorem{assumption}[theorem]{Assumption}
\newtheorem{definition}[theorem]{Definition}
\newtheorem{example}[theorem]{Example}
\newtheorem{conjecture}[theorem]{Conjecture}

\def\N{\mathcal{N}}
\def\R{\mathbb{R}}
\newcommand{\norm}[1]{\left\|#1\right\|}
\newcommand{\abs}[1]{\left|#1\right|}

\allowdisplaybreaks

\title{Distance preservation in state-space methods for detecting causal interactions in dynamical systems}
\author[1]{Matthew R.\ O'Shaughnessy\footnote{moshaughnessy6@gatech.edu}}
\author[1]{Mark A.\ Davenport}
\author[1]{Christopher J.\ Rozell}
\affil[1]{School of Electrical \& Computer Engineering, Georgia Institute of Technology, Atlanta, GA USA}
\date{\today}

\begin{document}

\maketitle

\begin{abstract}
	We analyze the popular ``state-space'' class of algorithms for detecting casual interaction in coupled dynamical systems. These algorithms are often justified by Takens' embedding theorem, which provides conditions under which relationships involving attractors and their delay embeddings are continuous. In practice, however, state-space methods often do not directly test continuity, but rather the stronger property of how these relationships preserve inter-point distances. This paper theoretically and empirically explores state-space algorithms explicitly from the perspective of distance preservation. We first derive basic theoretical guarantees applicable to simple coupled systems, providing conditions under which the distance preservation of a certain map reveals underlying causal structure. Second, we demonstrate empirically that typical coupled systems do not satisfy distance preservation assumptions. Taken together, our results underline the dependence of state-space algorithms on intrinsic system properties and the relationship between the system and the function used to measure it --- properties that are not directly associated with causal interaction.
\end{abstract}

\section{Introduction}
\label{sec:introduction}

Many topics of scientific inquiry involve the fundamental mathematical problem of using observational time series data to infer the causal structure underlying coupled dynamical systems \cite{runge2018causal,runge2019inferring,cliff2022unifying}. For example, determining regions of the brain in which seizures originate requires inferring causal relationships in neural activity that is frequently modeled by coupled dynamical systems \cite{levanquyen1999nonlinear}. Developing mechanistic understanding of climate systems involves detecting causal interactions in similar coupled dynamical system models \cite{runge2019inferring}. Analogous applications arise in fields from engineering and economics to political science and public policy.

Mathematically, suppose we have two potentially coupled dynamical systems with states $x$ and $y$:
\begin{equation}
	\begin{aligned}
	\dot{x} &= f(x,y) \\
	\dot{y} &= g(x,y).
	\end{aligned}
	\label{eq:forced-system}
\end{equation}
We say that $x$ drives $y$ ($x \to y$) if the differential equation governing $y$ has a nontrivial dependence on $x$, so that a perturbation to $x$ produces a change in the trajectory of $y$ \cite{pearl2009causality,peters2020causal}. Our goal is to determine whether $x \to y$ from the time series observations $\{x_t,y_t\}_{t=1}^T$.

The ``state-space'' category of algorithms infer causal structure in coupled dynamical systems using an idea, inspired by Takens' theorem, that has been termed the ``closeness principle'' \cite{amigo2018detecting}. These procedures aim to determine whether $x \to y$ from the relationship between inter-point distances on subsystems' delay embeddings. (See Sections \ref{sec:background-delay-embeddings-detecting-causal-structure} and \ref{sec:background-causal-inference} for a precise description.) But while these algorithms are typically based on how this relationship between delay embeddings \emph{preserves distances}, Takens' theorem guarantees only the weaker property of its \emph{continuity} (Figure \ref{fig:intro-fig}). There is thus a gap between the distance-preservation-based operation of state-space algorithms and the continuity-based theory they are inspired by, making it difficult for practitioners to know when state-space algorithms might be appropriate for a problem at hand.

\begin{figure}[t]
	\centering
	\includegraphics[width=0.6\textwidth]{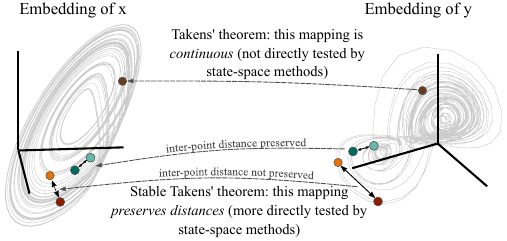}
	\caption{Takens' theorem guarantees the mapping between contemporaneous points on delay embeddings is diffeomorphic (continuous), but state-space methods are based on a stronger test of whether close distances are preserved between embeddings. This paper examines state-space methods through the lens of the stable Takens' theorem \cite{yap2011stable,eftekhari2018stabilizing}, which provides conditions under which distance preservation can be guaranteed.}
	\label{fig:intro-fig}
\end{figure}

In this paper we aim to reduce this gap by studying the closeness principle explicitly through the lens of distance preservation. We do so by making use of the recent \emph{stable} Takens' theorem \cite{yap2011stable,eftekhari2018stabilizing}, which shows that, under stronger conditions on the underlying dynamical systems, a delay embedding preserves not only \emph{topological} structure (which is not directly tested by closeness-principle-based methods), but also \emph{geometric} structure (which is more directly tested by closeness-principle-based methods). Our goal is to use the stable Takens' theorem to gain understanding about the settings under which the closeness principle holds and causal inference techniques based on it may be theoretically justified.

The main contributions of this paper are twofold. We first theoretically connect the causal structure of a system with properties of a key mapping in the restricted case of linear systems and observation functions (Section \ref{sec:theory-linear-phiy}), and provide a theoretically-grounded test for causal interaction based the closeness principle applicable to both linear and nonlinear systems (Section \ref{sec:theory-tests}). Our theoretical results build on the stable Takens' theorem of \cite{yap2011stable,eftekhari2018stabilizing} by applying them to heuristics used by state-space methods, and complement recent theoretical work on these algorithms that is based on continuity rather than distance preservation \cite{cummins2015efficacy}. Second, we provide empirical results that show the effectiveness of many state-space algorithms eludes straightforward distance-preservation-based explanations (Section \ref{sec:experiments}). \emph{The main message of this paper is that, while it is possible to provably identify causal interaction by analyzing how distances are preserved between delay embeddings, such tests are strongly dependent on system- and measurement-dependent properties that can vary significantly between systems and are often impossible to determine in practice.}

\section{Background}
\label{sec:background}

Sections \ref{sec:background-delay-embeddings-detecting-causal-structure} and \ref{sec:background-causal-inference} describe the delay embedding procedure, the closeness principle, and related tests for causal interaction. Sections \ref{sec:background-takens-theorem} and \ref{sec:background-stable-takens} then set the stage for our theoretical results with precise statements of Takens' theorem and the linear stable Takens' theorem.

\subsection{Delay embeddings and detecting causal interaction}
\label{sec:background-delay-embeddings-detecting-causal-structure}

Let the evolution of states $x \in \R^{n_x}$ and $y \in \R^{n_y}$ be described by dynamical systems $\dot{x} = f(x, y)$ and $\dot{y} = g(x, y)$. We say that $x$ drives $y$ ($x \to y$) if $g$ depends nontrivially on $x$ (so that a perturbation to $x$ will affect the trajectory of $y$), and similarly that $y$ drives $x$ ($y \to x$) if $f$ depends nontrivially on $y$ (so that a perturbation to $y$ will affect the trajectory of $x$). We use the notation $x \leftrightarrow y$ to indicate that both $x \to y$ and $y \to x$, and the notation $x \perp y$ to indicate that neither $x \to y$ nor $y \to x$.

In this paper we restrict our attention to systems that are governed by deterministic dynamics and constrained to a low-dimensional attractor after some time $t'$. We denote by $M_x$, $M_y$, and $M_{xy}$ the attractors of $x$, $y$, and $(x,y)$ such that (perhaps after a period of transient behavior) $x(t) \in M_x \subset \R^{n_x}$, $y(t) \in M_y \subset \R^{n_y}$, and $(x,y)(t) \in M_{xy} \subset \R^{n_x+n_y}$.\footnote{With slight abuse of notation, we refer alternately to an attractor itself and the set of points collected from it. For instance, we use $M_x$ to refer to both the attractor that $x(t)$ occupies and the set of observations $\{x_t\}$.} Throughout the paper we assume that data is collected once the system $(x,y)$ has reached its attractor --- that is, that the observed time series do not contain transients.

Using a scalar observation function $\varphi \colon M \to \R^m$, points on an attractor $M$ can be mapped to points on a \emph{delay embedding} in $N \subset \R^m$ as \cite{packard1980geometry}
\begin{equation}
	\Phi_{\varphi}(x_t,y_t) = \left[ \varphi(x_t,y_t),~ \varphi_x(x_{t-\tau},y_{t-\tau}),~ \dots~ \varphi(x_{t-(m-1)\tau},y_{t-(m-1)\tau}) \right]^T,
	\label{eq:delay-embedding}
\end{equation}
where the observation time lag $\tau$ and embedding dimension $m$ are parameters typically chosen to maximize some heuristic notion of information preservation \cite{fraser1986independent,casdagli1991state}.\footnote{Note that we have suppressed the dependence of $\Phi_{\varphi}$ on past time lags of $x_t$ and $y_t$.} We will be interested in observation functions with either $x$, $y$, or $(x,y)$ as their domains; for example,
\begin{equation*}
	\Phi_{\gamma_x}(x_t) = \left[ \gamma_x(x_t),~ \gamma_x(x_{t-\tau}),~ \dots~ \gamma_x(x_{t-(m-1)\tau}) \right]^T
\end{equation*}
uses the scalar observation function $\gamma_x \colon M_x \to \R$ to map $M_x$ to the delay embedding $N_x \subset \R^m$.

Takens' Theorem shows that, under conditions described in Section \ref{sec:background-takens-theorem}, there is a diffeomorphic relationship between an attractor $M$ and its delay embedding $N = \Phi_{\varphi}(M)$ \cite{packard1980geometry,takens1981detecting}. This implies that the mapping $\Phi_{\varphi} \colon M \to N$ is diffeomorphic, and thus continuous with continuous inverse.

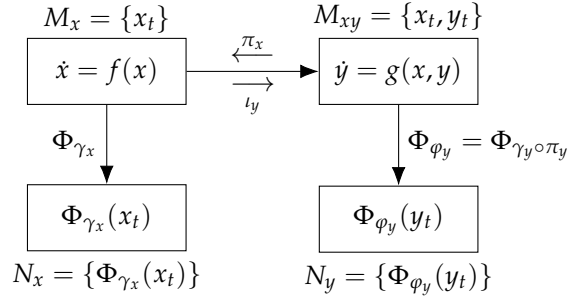
\begin{figure}[t]
	\centering
	\begin{tikzpicture}
		\matrix (m) [matrix of math nodes, row sep=3em, column sep=5em, minimum width=6em, minimum height=3em, nodes={draw=black, minimum height=2.5em}]
		{
			\dot{x} = f(x) & \dot{y} = g(x,y) \\
			\Phi_{\gamma_x}(x_t) & \Phi_{\varphi_y}(y_t) \\
		};
		\path[-{Latex[length=0.7em]}]
		(m-1-1) edge node [left] {$\Phi_{\gamma_x}$} (m-2-1)
		(m-1-2) edge node [right] {$\Phi_{\varphi_y}=\Phi_{\gamma_y\circ\pi_y}$} (m-2-2)
		(m-1-1) edge node [above] {$\overset{\pi_x}{\longleftarrow}$} node [below] {$\underset{\iota_y}{\longrightarrow}$} (m-1-2);
		\node[above=of m-1-1, xshift=0em, yshift=-3em] {$M_x = \{ x_t \}$};
		\node[above=of m-1-2, xshift=0em, yshift=-3em] {$M_{xy} = \{ x_t, y_t \}$};
		\node[below=of m-2-1, xshift=0em, yshift=2.8em] {$N_x = \{ \Phi_{\gamma_x}(x_t) \}$};
		\node[below=of m-2-2, xshift=0em, yshift=2.8em] {$N_y = \{ \Phi_{\varphi_y}(y_t) \}$};
	\end{tikzpicture}
	\caption{Attractors of dynamical systems and their delay-coordinate embeddings when $x \rightarrow y$. A critical distinction is the part of system $(x,y)$ that each observation function collects. Observation function $\gamma_x \colon M_x \to \R$ has only the $x$ subsystem as its domain. Observation function $\varphi_y \colon M_{xy} \to \R$, meanwhile, formally has the complete system $(x,y)$ as its domain, but depends nontrivially only on the $y$ subsystem. (That is, $\varphi_y$ can be written as the composition of a function $\gamma_y \colon M_y \to \R$ with the projection $\pi_y \colon (x,y) \mapsto y$.) We will be concerned with the conditions under which $\Phi_{\varphi_y}$ reveals information about the complete system $(x,y)$ despite collecting only information about $y$.}
	\label{fig:mappings}
\end{figure}

These relationships between attractors and their embeddings expose a property that state-space methods leverage to detect causal interactions \cite{cenys1991estimation,rulkov1995generalized,sugihara2012detecting}. Suppose that $x \to y$, so that the dynamical systems governing $x$ and $y$ can be written as $\dot{x} = f(x)$ and $\dot{y} = g(x,y)$ and the attractors $M_x = \{ x_t \}$ and $M_{xy} = \{ (x_t,y_t) \}$ are well-defined. With a fortuitous choice of observation functions, Takens' theorem then guarantees (see Section \ref{sec:background-takens-theorem}) that the delay mappings $\Phi_{\gamma_x} \colon M_x \to N_x$ and $\Phi_{\varphi_y} \colon M_{xy} \to N_y$ are bicontinous. Because the projection $\pi_x \colon M_{xy} \to M_x$ defined as $(x_t,y_t) \mapsto x_t$ is also continuous, there exists a continuous mapping from $N_y$ to $N_x$
\begin{equation*}
	N_y \overset{\Phi_{\varphi_y}^{-1}~}{\longrightarrow} M_{xy} \overset{\pi_x~}{\longrightarrow} M_x \overset{\Phi_{\gamma_x}~}{\longrightarrow} N_x
\end{equation*}
defined as $\Psi_{y \to x} = \Phi_{\gamma_x} \circ \pi_x \circ \Phi_{\varphi_y}^{-1}$ (Figure \ref{fig:mappings}) \cite{cummins2015efficacy,amigo2018detecting}. The converse, however, does not necessarily hold when $x \to y$: the corresponding mapping from $N_x$ to $N_y$, $\Psi_{x \to y} = \Phi_{\varphi_y} \circ \iota_y \circ \Phi_{\gamma_x}^{-1}$, involves the inclusion map $\iota_y = x_t \mapsto (x_t,y_t)$. Because $\iota_y$ is not (in general\footnote{Note, however, that when coupling from $x$ to $y$ becomes strong enough, the system can enter ``generalized synchrony'' \cite{rulkov1995generalized} in the sense that there exists a function $H$ such that $y(t) = H(x(t))$. In this case $y$ can also be predicted from $x$ and it becomes impossible to distinguish between the cases $x \to y$ and $y \to x$.}) continuous, when $x \to y$ the map $\Psi_{y \to x} \colon N_y \to N_x$ is continuous while $\Psi_{x \to y} \colon N_x \to N_y$ is not.

The fact that $\Psi_{y \to x}$ is continuous when $x \to y$, but (generally) not when $y \to x$ suggests a method for inferring causal structure from time series observations of $(x,y)$. However, it is difficult to determine the existence or nonexistence of this continuous mapping from data. In practice, heuristics based on \emph{relative distances} on $N_x$ and $N_y$ are commonly used \cite{rulkov1995generalized,quianquiroga2002performance,andrzejak2003bivariate,chicharro2009reliable,sugihara2012detecting,ma2015detecting,amigo2018detecting,krakovska2018comparison}. We state the general principle used by these methods as a conjecture:
\begin{conjecture}[``Closeness principle'']
	When $x \to y$ in the underlying system, pairs of points that are close to each other on $N_y$ will correspond to pairs of points that are also close on $N_x$. In notation, when $x \to y$ the map $\Psi_{y \to x} \colon N_y \to N_x$ preserves close distances.
	\label{conj:closeness-principle}
\end{conjecture}

\subsection{Closeness-principle-based methods for causal inference}
\label{sec:background-causal-inference}

A growing literature of state-space methods has applied versions of Conjecture \ref{conj:closeness-principle} using distance-based heuristics such as how distances are preserved between pairs of points on $N_y$ and their contemporaneous point pairs on $N_x$ \cite{rulkov1995generalized,arnhold1999robust,quianquiroga2002performance,andrzejak2003bivariate,chicharro2009reliable}, recurrence plots \cite{romano2007estimation,hirata2010identifying}, or more direct comparisons of corresponding distances on $N_y$ and $N_x$ \cite{cenys1991estimation,amigo2018detecting}. Here we outline a few representative approaches that are discussed further in Sections \ref{sec:experiments} and \ref{sec:discussion}.

One family of techniques operationalize a version of the closeness principle that is based on how nearest neighbor relationships are preserved.\footnote{While based on distance preservation, methods based on nearest neighbor preservation may be more robust to system properties such as observation noise. The preservation of nearest neighbors underlies other algorithms commonly applied to delay embeddings such as the false nearest neighbors method for determining embedding dimension \cite{kennel1992determining}} Denote the index of the $j^{\mathrm{th}}$ spatial nearest neighbor of $N_x[i]$ by $n_{ij}^{(x)}$, and the index of the $j^{\mathrm{th}}$ spatial nearest neighbor of $N_y[i]$ by $n_{ij}^{(y)}$.\footnote{In computing these nearest neighbor indices, the $W$ temporal neighbors (the ``Theiler window'') are often excluded.} Define $D_i(X) = \frac{1}{T'-1} \sum_{j=1,j \neq i}^{T'} \norm{N_x[i] - N_x[j]}_2^2$ to be the mean-square distance between points on $N_x$, $D_i^k(X) = \frac{1}{k} \sum_{j=1}^k \norm{N_x[i] - N_x[n_{ij}^{(x)}]}_2^2$ to be the mean-squared distance from a point on $N_x$ to its $k$ nearest neighbors, and $D_i^k(X \mid Y) = \frac{1}{k} \sum_{j=1}^k \norm{N_x[i] - N_x[n_{ij}^{(y)}]}_2^2$ to be the mean-squared distance from a point on $N_x$ to its $k$ ``mutual'' neighbors. Based on prior measures of \cite{arnhold1999robust,quianquiroga2002performance}, Andrzejak et al.\ propose
\begin{equation}
	M(X \mid Y) = \max \left\{ \frac{1}{T'} \sum_{i=1}^{T'} \frac{D_i(X) - D_i^k(X \mid Y)}{D_i(X) - D_i^k(X)},~ 0 \right\}
	\label{eq:metric-andrzejak2003bivariate}
\end{equation}
as a measure of directional influence of $x$ on $y$ \cite{andrzejak2003bivariate}. This metric reflects the intuition that as the coupling from $x$ to $y$ grows stronger, points that are spatially close on $N_y$ are more likely to correspond to points that are spatially close on $N_x$: when $x \perp y$, $D_i^k(X \mid Y) \approx D_i(X)$, but when $x \to y$, nearest neighbor relationships on $N_y$ become informative of nearest neighbor relationships on $N_x$ so that $D_i^k(X \mid Y) \approx D_i^k(X)$.\footnote{Empirical evidence has shown that $M(Y \mid X)$ also increases as the coupling $x \to y$ grows stronger, but not as quickly. The alternative metric $\Delta M = M(X \mid Y) - M(Y \mid X)$ has been suggested to mitigate this influence \cite{romano2007estimation,smirnov2005detection,chicharro2009reliable}. A measure that detects the overall (undirected) coupling strength, $M_s = \frac12 \left( M(X \mid Y) + M(Y \mid X) \right)$, has also be proposed \cite{smirnov2005detection}.} Statistical tests based on surrogate techniques are proposed in \cite{andrzejak2003bivariate}.

Rank-based versions of these measures have been proposed as variants that are more robust to attractor geometry and noise statistics. Define $g_{ij}$ to be the rank of the distance from $N_x[i]$ to $N_x[j]$, and let $G_i^k(X \mid Y) = \frac{1}{k} \sum_{j=1}^k g_{i,n_{ij}^{(y)}}$. Chicharro and Andrzejak \cite{chicharro2009reliable} propose
\begin{equation}
	L(X \mid Y) = \frac{1}{T'} \sum_{i=1}^{T'} \frac{G_i(X) - G_i^k(X \mid Y)}{G_i(X) - G_i^k(X)},
	\label{eq:metric-chicharro2009reliable}
\end{equation}
where $G_i(X) = \frac{T'}{2}$ is the mean rank between points in $N_x$ and $G_i^k(X) = \frac{k+1}{2}$ is the mean rank of a point to one of its $k$ nearest neighbors. A Wilcoxon signed rank test can be performed to evaluate the significance of $\Delta L = L(X \mid Y) - L(Y \mid X) > 0$.

The popular convergent cross-mapping (CCM) algorithm of \cite{sugihara2012detecting} evaluates how well $N_y$ can construct an estimate of $x$ using the mutual-neighbor-based ``simplex projection'' method of \cite{sugihara1990nonlinear}. Here, the estimate
\begin{equation*}
	\widehat{M_x}[i] \mid N_y = \sum_{j=1}^{m+1} w_j M_x[ n_{i,j}^{(y)} ]
\end{equation*}
is computed using the exponential weights
\begin{equation*}
	w_j = \frac{\exp\left(-\norm{N_y[t]-N_y[n_{i,j}^{(y)}]}/\norm{N_y[t]-N_y[n_{i,1}^{(y)}]}\right)}{\sum_{k=1}^{m+1} \exp\left(-\norm{N_y[t]-N_y[n_{i,k}^{(y)}]}/\norm{N_y[t]-N_y[n_{i,1}^{(y)}]}\right)}.
\end{equation*}
The existence of the relationship $x \to y$ is inferred based on whether the estimate $\widehat{M}_x \mid N_y$ increases in fidelity to $M_x$ when the attractor ``fills in'' as more points are used in the reconstruction. Variants of this technique improve inference by adding time delay information \cite{ye2015distinguishing}, using partial correlations to separate direct and indirect influences \cite{leng2020partial}, evaluating map smoothness using neural network accuracy \cite{ma2015detecting}, and compensating for short time series using neural network models of dynamics \cite{debrouwer2021latent}.

Other methods make more direct use of relative distances of point pairs on $N_y$ and their contemporaneous point pairs on $N_x$. In \cite{cenys1991estimation,romano2007estimation}, recurrence patterns --- measures based on the probability that $y$ returns to a neighborhood of $y_t$ when $x$ returns to the neighborhood of $x_t$, and vice-versa --- are used to assess the relative complexity of the dynamics of $x$ and $y$. In \cite{hirata2010identifying,amigo2018detecting}, causal structure is inferred from properties of the joint distribution of distances $d(x(t_1),x(t_2))$ and $d(y(t_1),y(t_2))$.

In perhaps the most rigorous approach, \cite{cummins2015efficacy} shows that causal relationships can be identified based on the existence of continuous noninjective surjective mappings between delay coordinate maps. To operationalize this theory, the authors apply the heuristic of \cite{pecora1995statistics} to detect the existence of continuous mappings between corresponding points on different delay coordinate maps. The approach of \cite{harnack2017topological} also examines the continuity and (non)injectivity of this mapping by quantifying the relative influence of a systems' intrinsic and external dynamics by examining the expansiveness of the mapping $\Psi_{y \to x} \colon N_y \to N_x$.

\subsection{Takens' theorem}
\label{sec:background-takens-theorem}

The guarantees provided by Takens' theorem \cite{takens1981detecting} are a key mathematical idea behind closeness-principle-based algorithms. Let $C^k(M,N)$ be the set of continuous functions mapping $M$ to $N$ with $k$ continuous derivatives, and let $D^k(M,N)$ be the set of diffeomorphisms mapping $M$ to $N$ (i.e., functions mapping $M$ to $N$ for which both the function and its inverse are in $C^k$).

\begin{theorem}[Takens' Theorem \cite{takens1981detecting}]
	Let $M$ be a compact manifold. For pairs $(\phi,\varphi)$ of dynamical system $\phi \in D^2(M,M)$ and measurement function $\varphi \in C^2(M,\R)$, the map $\Phi_{\varphi} \colon M \rightarrow \R^m$ defined in \eqref{eq:delay-embedding}, where $m = 2 \mathrm{dim}(M)+1$, is an embedding for generic $(\phi,\varphi) \subset D^2(M,M) \times C^2(M,\R)$.
	\label{thm:takens}
\end{theorem}

That $\Phi_{\varphi}(M)$ is an \emph{embedding} of $M$ means that the mapping $\Phi_{\varphi}$ is a homeomorphism: a continuous bijection with continuous inverse. That this holds for \emph{generic} pairs $(\phi,\varphi)$ means that the set of dynamical systems $\phi$ and measurement functions $\varphi$ for which $\Phi_{\varphi}$ produces an embedding of $M$ is open and dense in the $C^1$ topology. This genericity can be roughly thought of as telling us that ``most'' pairs $(\phi,\varphi)$ produce valid embeddings. The work of \cite{sauer1991embedology} reinforces this intuition by proving a version of the theorem that applies to a measure-theoretic \emph{prevalent} set of $(\phi,\varphi)$.

The \emph{forced} Takens' theorem extends this guarantee to the particular type of measurement functions and forced systems that we are often concerned about when detecting unidirectional coupling. Takens' theorem guarantees that a generic choice of $(\phi,\varphi)$ produces an embedding. However, the delay embeddings used by closeness-principle-based methods depend on \emph{fixed} measurement functions that in some cases use only information about either $x$ or $y$. (By using ``only information about $x$,'' we mean that a measurement function $\varphi_x$ can be written as $\varphi_x = \gamma_x \circ \pi_x$ for some $\gamma_x$, where $\pi_x$ is the projection $(x,y)(t) \mapsto x(t)$.) These fixed measurement functions do not necessarily belong to the open and dense set of $(\phi,\varphi)$ that satisfy Takens' theorem. The forced Takens' theorem of \cite{stark1999delay} shows that, under weak additional conditions on $\phi$, $\varphi_y \colon M_{xy} \to \R$ satisfies the conclusions of Takens' theorem when $x \to y$, and similarly that $\varphi_x \colon M_{xy} \to \R$ satisfies the conclusions of Takens' theorem when $y \to x$.

More formally, consider a forced dynamical system $(x,y) \subset M_x \times M_y$ of the form
\begin{align*}
	x_{t+1} &= f(x_t) \\
	y_{t+1} &= g(x_t,y_t)
\end{align*}
so that its underlying causal structure is $x \to y$. Takens' theorem holds for $(\phi,\varphi)$ in open and dense subsets of $D^2(M_x \times M_y, M_x \times M_y) \times C^2(M_x \times M_y, \R)$. However, if a delay embedding is constructed using only observations of $y$ --- as is done by state-space algorithms --- we would like to know when Takens' theorem holds for $(\phi,\varphi)$ in open and dense subsets of $D^2(M_x, M_x) \times D^2(M_x \times M_y, M_y) \times C^2(M_y, \R)$. The set $D^2(M_x, M_x) \times D^2(M_x \times M_y, M_y)$ is not generic in $D^2(M_x \times M_y, M_x \times M_y)$ and the  $C^2(M_y, \R)$ is not generic in $C^2(M_x \times M_y, \R)$, so Takens' theorem does not necessarily apply to this forced system. The \emph{forced} Takens' theorem of \cite{stark1999delay} provides (mild) additional assumptions under which the conclusions of Takens' theorem hold when only values of $y$ are observed:
\begin{theorem}[Forced Takens' Theorem {\cite[Thm.~3.1]{stark1999delay}}]
	Let $M_x$ and $M_{xy}$ be compact manifolds of dimension $n_x$ and $n_x+n_y \geq 1$, respectively. Suppose that the periodic orbits of period $<2d$ of $f \in D^2(M_x)$ are isolated and have distinct eigenvalues, where $d \geq (n_x+n_y) + 1$. Then for $r \geq 1$, there exists an open and dense set of $(f,\varphi) \in D^r(M_x \times M_y, M_x) \times C^r(M_x, \R)$ for which the map $\Phi_{f,g,\varphi}$ is an embedding.
	\label{thm:forced-takens}
\end{theorem}

For our purposes, the upshot of Takens' theorem and the forced Takens' theorem is that for ``most'' well-behaved $(\phi,\varphi)$ pairs --- even when the measurement function $\varphi$ collects information only from the $y$ subsystem --- the delay embedding mapping $\Phi_{\varphi}$ is a bijection, and both $\Phi_{\varphi}$ and $\Phi_{\varphi}^{-1}$ are continuous.

\subsection{Stable Takens' theorem}
\label{sec:background-stable-takens}

Although Takens' theorem shows that $\Phi_{\varphi}$ is a homeomorphism between an attractor $M$ and delay embedding $N$, it does not guarantee that this map preserves distances: points close to each other on $M$ may map to points that are far from each other on $N$ and vice-versa. This limitation is significant when analyzing closeness-principle-based state-space methods because in practice these algorithms rely on comparing how various mappings preserve distances between point pairs.

To understand the implications of this obstacle we employ the \emph{stable Takens' theorems} of \cite{yap2011stable,eftekhari2018stabilizing}, which show that, under more restrictive conditions on the measurement function and dynamical system, the delay map $\Phi_{\varphi}$ can preserve distances. Specifically, we say that the mapping $\Phi_{\varphi}$ is a \emph{stable} (i.e., distance-preserving) embedding of $M$ with conditioning $\delta \in [0,1)$ if for all distinct $p, q \in M$ and a constant $C > 0$,
\begin{equation}
	C (1-\delta) \leq \frac{\norm{\Phi_{\varphi}(p)-\Phi_{\varphi}(q)}_2^2}{\norm{p-q}_2^2} \leq C (1+\delta).
	\label{eq:stable-embedding}
\end{equation}
We refer to the lower and upper bounds on the amount by which $\Phi_{\varphi}$ can expand or contract distances between point pairs --- here, $C (1-\delta)$ and $C (1+\delta)$ --- as the lower and upper isometry constants.

For clarity, we begin by considering the simpler \emph{linear} stable Takens' theorem \cite{yap2011stable}. We then derive more general results that apply to nonlinear systems. The linear stable Takens' theorem applies to a class of linear dynamical systems observed with a linear measurement procedure:

\begin{definition}[Class $\mathcal{A}(d)$ of oscillatory linear dynamical systems \cite{yap2011stable}]
	A linear dynamical system with system matrix $A \in \R^{n \times n}$ is said to be in class $\mathcal{A}(d)$, $d \leq \frac{n}{2}$, if $A$ is (a) real, full rank, and has distinct eigenvalues; and (b) has only $d$ strictly imaginary conjugate pairs of eigenvalues and the remaining eigenvalues have strictly negative real components. The strictly imaginary conjugate pairs of eigenvalues are denoted $\{ \pm j \theta_i \}_{i=1}^d$, $\theta_1, \dots, \theta_d > 0$; their corresponding eigenvectors are denoted $v_1, v_1^*, \dots, v_d, v_d^*$. We define $\Lambda = \mathrm{diag}\{ j\theta_1, -j\theta_1, \dots, j\theta_d, -j\theta_d \}$ and $V = \left[v_1, v_1^*, \dots, v_d, v_d^* \right] \in \mathbb{C}^{n \times 2d}$ so that $A V = V \Lambda$.
	\label{def:class-Ad}
\end{definition}

Define $A_1 = \lambda_{\min}(V^H V)$ and $A_2 = \lambda_{\max}(V^H V)$ to be the minimum and maximum eigenvalues of $V^H V$. Using a linear measurement function $\varphi(\cdot) = \left< h, \cdot \right> \colon \R^n \to \R$, denote the minimum and maximum alignment of system eigenvectors with the measurement function as $\kappa_1 = \min_{i = 1, \dots, d}~ \frac{1}{\norm{h}_2} \abs{v_i^H h}$ and $\kappa_2 = \max_{i = 1, \dots, d}~ \frac{1}{\norm{h}_2} \abs{v_i^H h}$. Finally, define
\begin{gather*}
	\nu = \max_{i \neq j} \biggl\{ \abs{\sin(\theta_i T_s)}^{-1},~ \abs{\sin\left(\frac{(\theta_i-\theta_j) T_s}{2}\right)}^{-1},~ \abs{\sin\left(\frac{(\theta_i+\theta_j) T_s}{2}\right)}^{-1} \biggr\}.
\end{gather*}

\begin{theorem}[Stable Takens' theorem for linear systems \cite{yap2011stable}]
	Let there be a linear dynamical system defined by system matrix $A \in \R^{n \times n}$ of class $\mathcal{A}(d)$ that has reached its attractor, a sampling interval $T_s > 0$, and a measurement function $\varphi(\cdot) = \left<h,\cdot\right> \colon \R^n \to \R$ such that $\norm{h}_2^2 = \frac{2d}{m}$. Suppose that (a) $m > (2d-1) \frac{A_2 \kappa_2^2}{A_1 \kappa_1^2} \nu$; (b) $\left\{ \pm j \theta_i \right\}$ are distinct and strictly complex; and (c) $v_i^H h \neq 0$ for all $i = 1, \dots, d$. Then for all distinct $p, q \in M$, the delay embedding map $\Phi_{\varphi}$ satisfies \eqref{eq:stable-embedding} with $C = d\left( \frac{\kappa_1^2}{A_2} + \frac{\kappa_2^2}{A_1} \right)$ and $\delta = \delta_0 + \delta_1(m)$, where
	\begin{equation*}
		\delta_0 = \frac{A_2 \kappa_2^2 - A_1 \kappa_1^2}{A_2 \kappa_2^2 + A_1 \kappa_1^2} \quad\text{and}\quad \delta_1(m) = \frac{(2d-1) \nu}{m} \left( \frac{2 A_2 \kappa_2^2}{A_2 \kappa_2^2 + A_1 \kappa_1^2} \right).
	\end{equation*}
	\label{thm:linear-stable-takens}
\end{theorem}

\section{Theoretical results}
\label{sec:theory}

In this section, we provide analytical results related to the closeness principle. We begin in Section \ref{sec:theory-linear-phiy} by focusing on the key mapping $\Phi_{\varphi_y} \colon M_{xy} \to N_y$ (see Figure \ref{fig:mappings-all}) in the restricted class of linear systems $\mathcal{A}(d)$ described by Definition \ref{def:class-Ad}. Using both a proof of concept and a formalized result, we show that the stability of $\Phi_{\varphi_y}$ depends on the causal structure relating $x$ and $y$. Section \ref{sec:theory-tests} then builds on this intuition, developing a test for causal interaction based on the expansivity of $\Psi_{y \to x}$ for general (potentially nonlinear) systems.

\subsection{Causal structure affects stability of $\varphi_y$ in linear systems}
\label{sec:theory-linear-phiy}

We begin by exploring the mapping $\Phi_{\varphi_y} \colon M_{xy} \to N_y$ (see Figure \ref{fig:mappings-all}) in the restricted class of linear systems $\mathcal{A}(d)$. Specifically, we show that $\Phi_{\varphi_y}$ cannot be stable when $x \not\to y$. This formalizes the key intuition that a delay embedding of $M_{xy}$ constructed with only information about $y$ can only be stable when the lags of $y$ ``contain information about the full system'' because $x \to y$.

In this subsection we consider the class of linear coupled systems in which $x \to y$ but $y \not\to x$,
\begin{equation}
	\begin{bmatrix} \dot{x} \\ \dot{y} \end{bmatrix} = \begin{bmatrix} A_{xx} & 0 \\ A_{yx} & A_{yy} \end{bmatrix} \begin{bmatrix} x \\ y \end{bmatrix}.
	\label{eq:linear-forced-system}
\end{equation}
Here $A_{xx} \in \R^{n_x \times n_x}$ describes the internal dynamics of $x$, $A_{yy} \in \R^{n_y \times n_y}$ describes the internal dynamics of $y$, and $A_{yx} \in \R^{n_y \times n_x}$ describes the external dynamics of $y$ due to $x$'s forcing.

\subsubsection{Example: A simple system analytically satisfies the closeness principle}
\label{sec:theory-linear-phiy-example}

\begin{figure}[t]
	\centering
	\begin{tikzpicture}
		\matrix (m) [matrix of math nodes, row sep=4em, column sep=6em, minimum width=4em, minimum height=3em, nodes={draw=black, minimum height=2.5em}] { M_x & M_{xy} & M_y \\ N_x & & N_y \\ };
		\path[-{Latex[length=0.7em]}]
		(m-1-1) edge node [left] {$\Phi_{\gamma_x}$} (m-2-1)
		(m-1-3) edge node [right] {$\Phi_{\gamma_y}$} (m-2-3)
		(m-1-2) edge node [left] {$\Phi_{\varphi_x}$} (m-2-1)
		(m-1-2) edge node [right] {$\Phi_{\varphi_y}$} (m-2-3)
		(m-1-2) edge[transform canvas={yshift=1mm}] node [above] {$\pi_x$} (m-1-1)
		(m-1-1) edge[transform canvas={yshift=-1mm}] node [below] {$\iota_y$} (m-1-2)
		(m-1-2) edge[transform canvas={yshift=1mm}] node [above] {$\pi_y$} (m-1-3)
		(m-1-3) edge[transform canvas={yshift=-1mm}] node [below] {$\iota_x$} (m-1-2)
		(m-2-3) edge [dashed, transform canvas={yshift=1mm}] node [above] {$\Psi_{y \to x}$} (m-2-1)
		(m-2-1) edge [dashed, transform canvas={yshift=-1mm}] node [below] {$\Psi_{x \to y}$} (m-2-3);
	\end{tikzpicture}
	\caption{Mappings between attractors and delay embedding using the observation functions we consider.}
	\label{fig:mappings-all}
\end{figure}
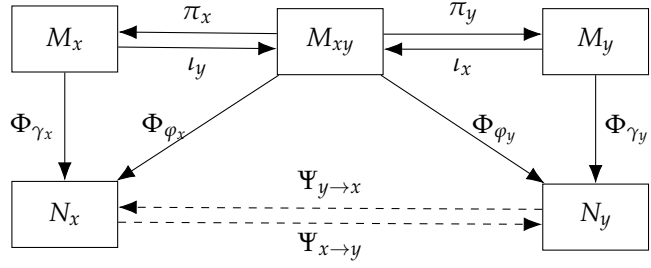

The following example constructs a forced linear system with linear measurement functions by extending a simple autonomous system from \cite{yap2011stable}:

\begin{example}[Well-behaved forced linear system in which $x \to y$]
	Let $A = V \Lambda V^{-1} \in \R^{4 \times 4}$ be the system matrix of a forced system in $\mathcal{A}(d=2)$ [see \eqref{eq:linear-forced-system}] with matrix of eigenvectors
	\begin{equation*}
		V = \begin{bmatrix} \frac{1}{\sqrt{2}} v_{xx} & \frac{1}{\sqrt{2}} v_{xx}^* & 0 & 0 \\ \frac{1}{\sqrt{2}} v_{yx} & \frac{1}{\sqrt{2}} v_{yx}^* & v_{yy} & v_{yy}^* \end{bmatrix},
	\end{equation*}
	where $v_{xx} = v_{yx} = v_{yy} = \frac{1}{\sqrt{2}} \begin{bmatrix} 1 \\ j \end{bmatrix}$ and $j = \sqrt{-1}$; matrix of eigenvalues
	\begin{equation*}
		\Lambda = \mathrm{diag}\left( \left[ j \theta_1,~ -j \theta_1,~ j \theta_2,~ -j \theta_2 \right] \right),
	\end{equation*}
	where $\theta_1 = 2.3129$, $\theta_2 = 0.1765$; and initial condition $(x_0,y_0) = V \cdot 1_{4 \times 1}$. (This initial condition ensures that the system does not exhibit transient behavior.)
	
	As in \cite{yap2011stable} we use linear measurement functions $\varphi$ of the form $\varphi(\cdot) = \left< h, \cdot \right>$, with vectors $h$ constructed to approximately align with the eigenvectors of $A$ as
	\begin{gather*}
		c_{\varphi_{xy}} = (1+r_1) \mathrm{Re}(V_1) + (1+r_2) \mathrm{Im}(V_1) + (1+r_3) \mathrm{Re}(V_3) + (1+r_4) \mathrm{Im}(V_3),
	\end{gather*}
	where $V_1$ and $V_3$ denote the first and third columns of $V$, and $r_i \sim \N(0,0.1)$ for $i = 1, \dots, 4$. We construct normalized $h_{\varphi_{xy}} \in \R^{n_x+n_y}$, $h_{\varphi_x} \in \R^{n_x+n_y}$, $h_{\varphi_y} \in \R^{n_x+n_y}$, $h_{\gamma_x} \in \R^{n_x}$, and $h_{\gamma_y} \in \R^{n_y}$ by first extracting or setting to zero the appropriate parts of $c_{\varphi_{xy}}$,
	\begin{gather*}
		c_{\varphi_x} = \begin{bmatrix} c_{\varphi_{xy}}[1] \\ c_{\varphi_{xy}}[2] \\ 0 \\ 0 \end{bmatrix},~
		c_{\varphi_y} = \begin{bmatrix} 0 \\ 0 \\ c_{\varphi_{xy}}[3] \\ c_{\varphi_{xy}}[4] \end{bmatrix}, ~
		c_{\gamma_x} = \begin{bmatrix} c_{\varphi_{xy}}[1] \\ c_{\varphi_{xy}}[2] \end{bmatrix},~
		c_{\gamma_y} = \begin{bmatrix} c_{\varphi_{xy}}[3] \\ c_{\varphi_{xy}}[4] \end{bmatrix},
	\end{gather*}
	and then normalizing so that
	\begin{gather*}
		h_{\varphi_{xy}} = \sqrt{\frac{4}{m}} \frac{c_{\varphi_{xy}}}{\norm{c_{\varphi_{xy}}}_2},~
		h_{\varphi_x} = \sqrt{\frac{4}{m}} \frac{c_{\varphi_x}}{\norm{c_{\varphi_x}}_2},~
		h_{\varphi_y} = \sqrt{\frac{4}{m}} \frac{c_{\varphi_y}}{\norm{c_{\varphi_y}}_2},~
		h_{\gamma_x} = \sqrt{\frac{2}{m}} \frac{c_{\gamma_x}}{\norm{c_{\gamma_x}}_2},~
		\text{and}~ h_{\gamma_y} = \sqrt{\frac{2}{m}} \frac{c_{\gamma_y}}{\norm{c_{\gamma_y}}_2}.
	\end{gather*}
	These measurement functions are used to construct the delay embeddings shown in Figure \ref{fig:mappings-all}: $N_{xy} = \Phi_{\varphi_{xy}}(M_{xy})$, $N_x = \Phi_{\varphi_x}(M_{xy}) = \Phi_{\gamma_x}(M_x)$, and $N_y = \Phi_{\varphi_y}(M_{xy}) = \Phi_{\gamma_y}(M_y)$. Each embedding is created using $m$ time delays of $\tau = 1$.
	\label{ex:linear-forced-system}
\end{example}

The system of Example \ref{ex:linear-forced-system} consists of an autonomous subsystem $x$ and a forced subsystem $y$. For the autonomous subsystem $x$ defined by system matrix $A_{xx}$, a straightforward computation reveals that the relevant quantities in Theorem \ref{thm:linear-stable-takens} are $C = 1$ and $\lim_{m \to \infty} \delta(m) = 0$. Theorem \ref{thm:linear-stable-takens} thus guarantees that the embedding $\Phi_{\gamma_x} \colon M_x \to N_x$ is an exact isometry as $m \to \infty$. Theorem \ref{thm:linear-stable-takens} can similarly be used to bound the stability of $\Phi_{\varphi_{xy}}$, $\Phi_{\varphi_x}$, and $\Phi_{\varphi_y}$. (The stability of $\Phi_{\gamma_y}$ cannot be bounded analytically because the $y$ subsystem cannot be separated from the driving $x$ subsystem.)

\begin{table*}
	\centering
\begin{tabular}{@{}ccccccccc@{}}
\toprule
& ~ & \multicolumn{2}{c}{Lower isometry constant} & \multicolumn{2}{c}{Upper isometry constant} & ~ & \multicolumn{2}{c}{Ratio of isometry constants} \\
\cmidrule(lr){3-4}\cmidrule(lr){5-6}\cmidrule(lr){8-9}Mapping & & Theoretical & Empirical & Theoretical & Empirical & & Theoretical & Empirical \\
\midrule
$\Phi_{\varphi_{xy}} \colon M_{xy} \to N_{xy}$ & & 0.60 & 1.00 & 6.23 & 5.83 & & 10.35 & 5.84 \\
$\Phi_{\varphi_x} \colon M_{xy} \to N_x$ & & \textless10\textsuperscript{-3} & \textless10\textsuperscript{-3} & 6.23 & 2.00 & & \textgreater10\textsuperscript{3} & \textgreater10\textsuperscript{3} \\
$\Phi_{\varphi_y} \colon M_{xy} \to N_y$ & & 0.12 & 0.76 & 7.30 & 5.24 & & 61.25 & 6.87 \\
$\Phi_{\gamma_x} \colon M_x \to N_x$ & & 0.99 & 1.00 & 1.01 & 1.00 & & 1.01 & 1.00 \\
$\Phi_{\gamma_y} \colon M_y \to N_y$ & & -- & 0.50 & -- & \textgreater10\textsuperscript{3} & & -- & \textgreater10\textsuperscript{3} \\
$\Psi_{x \to y} \colon N_x \to N_y$ & & 0.12 & 1.00 & -- & \textgreater10\textsuperscript{3} & & -- & \textgreater10\textsuperscript{3} \\
$\Psi_{y \to x} \colon N_y \to N_x$ & & -- & \textless10\textsuperscript{-3} & 8.44 & 1.00 & & -- & \textgreater10\textsuperscript{3} \\
\bottomrule
\end{tabular}

	\caption{Analytical and empirical isometry constants for the linear forced system of Example \ref{ex:linear-forced-system} and $m = 250$.}
	\label{tab:linear-forced-system-isoconst}
\end{table*}

To compare these theoretical bounds to their empirical values, we simulate Example \ref{ex:linear-forced-system} for 10,000 time steps and compute the isometry constants for the maps in Figure \ref{fig:mappings-all} using 50,000 randomly sampled point pairs $(p,q)$. To construct delay embeddings we use $m = 250$ time delays so that the isometry constants are limited more by structural properties of the system than by an insufficient value of $m$. These analytical bounds and empirically computed isometry constants for the system in Example \ref{ex:linear-forced-system} and maps between attractors and delay embeddings shown in Figure \ref{fig:mappings-all} are shown in Table \ref{tab:linear-forced-system-isoconst}. Observe that, as expected, the mapping of the autonomous subsystem $\Phi_{\gamma_x}$ becomes an exact isometry as $m \to \infty$, while the mapping of the forced subsystem $\Phi_{\gamma_y}$ does not.

A key observation about this system is that the stability of the mappings from $M_{xy}$ to $N_x$ and $N_y$ reveals information about underlying causal structure: the mapping $\Phi_{\varphi_y} \colon M_{xy} \to N_y$, which informally produces a delay embedding that ``contains information about both $x$ and $y$'' because $x \to y$, is stable \emph{despite being constructed from only measurements of the $y$ subsystem}. However, the mapping $\Phi_{\varphi_x} \colon M_{xy} \to N_x$, which informally produces a delay embedding that ``contains information only about $x$,'' is not stable.

This intuition carries to the relative stability of the contemporaneous point mappings $\Psi_{x \to y} \colon N_x \to N_y$ and $\Psi_{y \to x} \colon N_y \to N_x$, which are used by algorithms described in Section \ref{sec:background-causal-inference}. Table \ref{tab:linear-forced-system-isoconst} shows that the intuition of the closeness principle indeed holds in this simple example where $x \to y$. The map $\Psi_{y \to x}$ has bounded \emph{expansiveness}: pairs of points that are close on $N_y$ are also close on $N_x$. By contrast, $\Psi_{x \to y}$ has bounded \emph{contractiveness}: pairs of points that are close on $N_y$ may be far apart on $N_x$.

\subsubsection{Causal structure affects stability of $\varphi_y$}
\label{sec:theory-linear-phiy-result}

We next formalize the intuition that the stability of $\Phi_{\varphi_y}$ depends on causal structure in the restricted setting of functions in $\mathcal{A}(d)$. The proof demonstrates that while $\Phi_{\varphi_y}$'s stability indeed depends on whether or not $x \to y$ in the underlying system, it is also impacted by properties of the system matrix and its relationship to the measurement operator $\varphi_y$.

\begin{proposition}
	Let $(x,y)$ be a forced linear dynamical system [i.e., of form \eqref{eq:linear-forced-system}] in class $\mathcal{A}(d)$. Denote this system's attractor by $M_{xy}$, and let its delay embedding $N_y = \Phi_{\varphi_y}(M_{xy})$ be constructed with a linear observation function $\varphi_y: M_{xy} \to \R$ that ``uses only information about $y$.'' (That is, let the measurement function $\varphi_y$ take the form $\varphi_y(\cdot) = \left< h_{\gamma_y}, \pi_y(\cdot) \right>$, where $\pi_y$ is the projection $(x,y) \mapsto y$ and $h_{\gamma_y} \in \R^{n_y}$.) Then:
	\begin{itemize}
		\item When $x \not\to y$ in the underlying system, $\Phi_{\varphi_y}$ is not a stable map.
		\item When $x \to y$ in the underlying system, $\Phi_{\varphi_y}$ may be a stable map.
	\end{itemize}
	\label{thm:linear-forced-stable-takens}
\end{proposition}

The proof is located in Appendix \ref{sup:proofs-linear-forced-stable-takens}. Proposition \ref{thm:linear-forced-stable-takens} makes precise the intuition described in Example \ref{ex:linear-forced-system}: when constructing a delay embedding using measurements of only one variable of a dynamical system, geometry is preserved only when that one variable ``contains information about'' the complete system. That is, in a system $(x,y)$ where $x \to y$, a delay embedding constructed with only observations of $y$ may be geometry-preserving, but one constructed with only observations of $x$ cannot be.

\subsection{Tests for causal interaction based on expansivity of $\Psi_{y \to x}$}
\label{sec:theory-tests}

We next move closer to examining the map $\Psi_{y \to x}$ that is considered by the closeness principle. The results in this subsection, outlined in Figure \ref{fig:theory-overview}, relate the inter-point distances on various embeddings to causal interaction. The results are stated in terms of the distance preservation of the mappings shown in Figure \ref{fig:mappings-all}; recall that when discussing these mappings we assume that they are operating on systems that have reached their attractors (i.e., that observed time series do not contain transients). The results in this section are not restricted to linear systems.

\begin{figure}[t]
	\centering
	\includegraphics[width=0.7\textwidth]{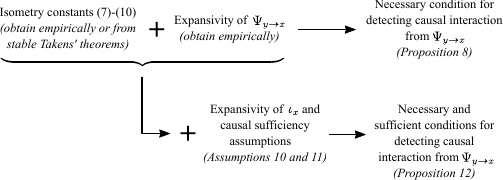}
	\caption{Overview of theoretical results in Section \ref{sec:theory-tests} relating distance preservation \eqref{eq:stabletakens-Phi-gammax}--\eqref{eq:stabletakens-Phi-phiy} to conditions for detecting causal interaction from $\Psi_{y \to x}$.}
	\label{fig:theory-overview}
\end{figure}

\subsubsection{Necessary conditions}

Suppose we have analytical or empirical bounds on the stability of the maps $\Phi_{\gamma_x}$, $\Phi_{\gamma_y}$, $\Phi_{\varphi_x}$, and $\Phi_{\varphi_y}$ (see Figure \ref{fig:mappings-all}) in hand:
\begin{align}
	l_{\gamma_x} &\leq \frac{\norm{\Phi_{\gamma_x}(p)-\Phi_{\gamma_x}(q)}_2^2}{\norm{p-q}_2^2} \leq u_{\gamma_x} \quad\forall~ p,q \in M_x, \label{eq:stabletakens-Phi-gammax} \\
	l_{\gamma_y} &\leq \frac{\norm{\Phi_{\gamma_y}(p)-\Phi_{\gamma_y}(q)}_2^2}{\norm{p-q}_2^2} \leq u_{\gamma_y} \quad\forall~ p,q \in M_y, \label{eq:stabletakens-Phi-gammay} \\
	l_{\varphi_x} &\leq \frac{\norm{\Phi_{\varphi_x}(p)-\Phi_{\varphi_x}(q)}_2^2}{\norm{p-q}_2^2} \leq u_{\varphi_x} \quad\forall~ p,q \in M_{xy}, \label{eq:stabletakens-Phi-phix} \\
	l_{\varphi_y} &\leq \frac{\norm{\Phi_{\varphi_y}(p)-\Phi_{\varphi_y}(q)}_2^2}{\norm{p-q}_2^2} \leq u_{\varphi_y} \quad\forall~ p,q \in M_{xy}. \label{eq:stabletakens-Phi-phiy}
\end{align}

Many tests in the literature (see Section \ref{sec:background-causal-inference}) are stated in terms of properties of the map between delay embeddings $\Psi_{y \to x} \colon N_y \to N_x$. The following result provides a test for causal relationships using the stability of this map:

\begin{proposition}[Necessary condition for detecting causal interaction using the stable Takens' theorem]
	Suppose the measurement operators $\gamma_x \colon M_x \to N_x$, $\gamma_y \colon M_y \to N_y$, $\varphi_x \colon M_{xy} \to N_x$, and $\varphi_y \colon M_{xy} \to N_y$ are in the generic sets that satisfy Takens' theorem for systems $x$, $y$, and $(x,y)$. Further suppose these measurement functions applied to these dynamical systems satisfy \eqref{eq:stabletakens-Phi-gammax}--\eqref{eq:stabletakens-Phi-phiy} with isometry constants $l_{\varphi_x}$, $l_{\varphi_y}$, $u_{\gamma_x}$, and $u_{\gamma_y}$. Then:
	\begin{itemize}
		\item If $x \to y$, then
		\begin{equation}
			\frac{\norm{\Psi_{y \to x}(p) - \Psi_{y \to x}(q)}_2^2}{\norm{p-q}_2^2} \leq \frac{u_{\gamma_x}}{l_{\varphi_y}}
		\end{equation}
		for all $p, q \in N_y$.
		\item If $y \to x$, then
		\begin{equation}
			\frac{\norm{\Psi_{y \to x}(p) - \Psi_{y \to x}(q)}_2^2}{\norm{p-q}_2^2} \geq \frac{l_{\varphi_x}}{u_{\gamma_y}}
		\end{equation}
		for all $p, q \in N_y$.
	\end{itemize}
	\label{thm:causality-stable-takens-1}
\end{proposition}

The proof, which applies the isometry constants to the compositions of mappings that can comprise $\Psi_{y \to x}$, is located in Appendix \ref{sup:proofs-causality-stable-takens-1}. The result admits the following test for causal interaction as a simple corollary:
\begin{corollary}[Expansivity of $\Psi_{y \to x}$ implies $x \not\to y$]
	Suppose dynamical system $(x,y)$ and the measurement functions used to construct its delay embeddings satisfy the conditions of Proposition \ref{thm:causality-stable-takens-1}. If $\exists$ $p, q \in N_y$ such that
	\begin{equation*}
		\frac{\norm{\Psi_{y \to x}(p) - \Psi_{y \to x}(q)}_2^2}{\norm{p-q}_2^2} > \frac{u_{\gamma_x}}{l_{\varphi_y}},
	\end{equation*}
	then $x \not\to y$.
	\label{thm:causality-stable-takens-1-test}
\end{corollary}

\begin{proof}
	Contrapositive of Proposition \ref{thm:causality-stable-takens-1}.
\end{proof}

Corollary \ref{thm:causality-stable-takens-1-test} result provides a theoretically-grounded test for causality based on distance preservation. It suggests empirically or theoretically computing the isometry constants $u_{\gamma_x}$ and $l_{\varphi_y}$, then attempting to find a pair of points such that the ratio of their squared distance on $N_y$ and the contemporaneous squared distance on $N_x$ is larger than $u_{\gamma_x} / l_{\varphi_y}$. Such a pair of points, if found, serves as a certificate guaranteeing that $x \not\to y$.

However, Corollary \ref{thm:causality-stable-takens-1-test} provides only a method to guarantee the \emph{nonexistence} of the causal link $x \to y$. This reflects the proper interpretation of many state-space causal inference algorithms such as CCM: for these tests, the existence of a causal relationship $x \to y$ is a necessary, but not sufficient, condition for $\Psi_{y \to x}$ to have certain properties \cite{sugihara2012detecting,krakovska2018comparison}. In most applications, however, it would be helpful to guarantee the \emph{existence} of causal links. We next show that this is possible under two strong assumptions.

\subsubsection{Necessary and sufficient conditions}

The result below strengthens Proposition \ref{thm:causality-stable-takens-1} to show that $x \to y$ is both a necessary and sufficient condition for $\Psi_{y \to x}$ to preserve distances. Showing this requires two assumptions. The first compensates for the looseness of the near-isometries guaranteed by \eqref{eq:stabletakens-Phi-gammax}--\eqref{eq:stabletakens-Phi-phiy}:
\begin{assumption}[Expansivity of $\iota_x$ when $y \to x$]
	\label{ass:inclusion-map-expansivity}
	If $y \to x$, then $\iota_x \colon M_y \to M_{xy}$ satisfies
	\begin{equation*}
		\frac{\norm{\iota_x(p) - \iota_x(q)}_2^2}{\norm{p-q}_2^2} > \frac{u_{\gamma_x} u_{\gamma_y}}{l_{\varphi_x} l_{\varphi_y}}
	\end{equation*}
	for \emph{some} $p,q \in M_y$.
\end{assumption}

Assumption \ref{ass:inclusion-map-expansivity} states that when $y \to x$ there must be a pair of points on $M_{xy}$ that become sufficiently closer together when projected onto $M_y$. To gain intuition, take any two points $\{y_t,y_s\}$ on $M_y$ and consider their contemporaneous points on $M_x$, $\{x_t,x_s\}$. Assumption \ref{ass:inclusion-map-expansivity} states that for $t \neq s$,
\begin{equation*}
	\norm{x_t-x_s}_2^2 > \left( \frac{u_{\gamma_x} u_{\gamma_y}}{l_{\varphi_x} l_{\varphi_y}} - 1 \right) \norm{y_t-y_s}_2^2,
\end{equation*}
that is, that the distance between $x_t$ and $x_s$ is sufficiently large with respect to the distance between $y_t$ and $y_s$. 

The second assumption eliminates the possibility that $\Psi_{y \to x}$ is stable ``by coincidence'':
\begin{assumption}[Causal sufficiency]
	\label{ass:causal-sufficiency}
	If $\exists$ $p,q \in N_y$ such that
	\begin{equation*}
		\frac{\norm{\Psi_{y \to x}(p)-\Psi_{y \to x}(q)}_2^2}{\norm{p-q}_2^2} \leq \frac{u_{\gamma_x}}{l_{\varphi_y}},
	\end{equation*}
	then there is a causal connection between the dynamical systems, either $x \to y$ or $y \to x$. In other words, if $\Psi_{y \to x}$ is a stable mapping we cannot have $x \perp y$.
\end{assumption}

Assumption \ref{ass:causal-sufficiency} is a type of causal sufficiency assumption; it eliminates the possibility that the systems $x$ and $y$ are not directly related but share structure by coincidence or because of an external system. For instance, Assumption \ref{ass:causal-sufficiency} precludes the existence of an external confounder $z$ that affects both $x$ and $y$, or a setting in which $x$ and $y$ are two independent yet identical systems started at the same time. The strength of this assumption is the reason that many techniques that rely on observational data do not claim to be able to detect ``causal'' interactions.

Assumptions \ref{ass:inclusion-map-expansivity} and \ref{ass:causal-sufficiency} allow us to strengthen Proposition \ref{thm:causality-stable-takens-1}:

\begin{proposition}[Necessary and sufficient conditions for detecting causal interaction using the stable Takens' theorem]
	\label{thm:causality-stable-takens-iff}
	Suppose that $(x,y)$ and its measurement functions satisfy Assumption \ref{ass:inclusion-map-expansivity}, Assumption \ref{ass:causal-sufficiency}, and the conditions of Proposition \ref{thm:causality-stable-takens-1}. Then $x \to y$ if and only if
	\begin{equation*}
		\frac{\norm{\Psi_{y \to x}(p) - \Psi_{y \to x}(q)}_2^2}{\norm{p-q}_2^2} \leq \frac{u_{\gamma_x}}{l_{\varphi_y}}
	\end{equation*}
	for all $p,q \in N_y$.
\end{proposition}

The proof, which applies the isometry constants \eqref{eq:stabletakens-Phi-gammax}--\eqref{eq:stabletakens-Phi-phiy} to the compositions of mappings that can comprise $\Psi_{y \to x}$ using the additional restrictions of Assumptions \ref{ass:inclusion-map-expansivity} and \ref{ass:causal-sufficiency}, is located in Appendix \ref{sup:proofs-causality-stable-takens-iff}. This test suggests empirically or theoretically computing the isometry constants $u_{\gamma_x}$ and $l_{\varphi_y}$, then attempting to find a pair of points such that the ratio of their squared distance on $N_y$ and the contemporaneous squared distance on $N_x$ is larger than $u_{\gamma_x} / l_{\varphi_y}$. If Assumptions \ref{ass:inclusion-map-expansivity} and \ref{ass:causal-sufficiency} are satisfied, then Proposition \ref{thm:causality-stable-takens-iff} allows a purely distance-based test to be used to guarantee either the existence or nonexistence of $x \to y$.

\subsection{Discussion of implications}
\label{sec:theory-discussion}

The results in this section suggest the following for algorithm developers and practitioners:
\begin{itemize}
    \item By showing that causal structure impacts \emph{whether} $\Phi_{\varphi_y}$ preserves inter-point distances, Proposition \ref{thm:linear-forced-stable-takens} provides some theoretical evidence in support of using the closeness principle to detect causal structure. However, this result depends on the isometry constants \eqref{eq:stabletakens-Phi-gammax}--\eqref{eq:stabletakens-Phi-phiy}, which in turn depend on a multitude of other system properties (as illustrated by the stable Takens' theorem (Theorem \ref{thm:linear-stable-takens})). This supports previous empirical observations \cite{cobey2016limits,krakovska2019correlation} that \emph{how well} distances are preserved depends on properties of the system under study (see Section \ref{sec:discussion} for details). In addition to the system properties highlighted in previous work (see Section \ref{sec:discussion}), our results suggest that ``alignment'' \footnote{More precisely, by ``alignment'' in the linear stable Takens' theorem of \cite{yap2011stable} we mean the minimum and maximum values of $\{ \kappa_i = \frac{1}{\norm{h}_2} \abs{v_i^H h} \}$ (see Definition \ref{def:class-Ad} and Theorem \ref{thm:linear-stable-takens}). In the general stable Takens' theorem of \cite{eftekhari2018stabilizing} we mean the ``stable rank of the attractor'' defined in \cite[Sec.~3]{eftekhari2018stabilizing}, which captures alignment of the measurement function and system. Intuitively, this means that inter-point distances can be guaranteed to be better preserved when the measurement functions more fully capture information about the system (see \cite[Remark 3.2]{eftekhari2018stabilizing}).} of a system and the measurement functions used to construct delay embeddings could impact the results of closeness-principle-based algorithms. This indicates an important sensitivity of these methods to how data is collected.
    \item Proposition \ref{thm:linear-forced-stable-takens} deals with a restricted class of linear systems and measurement functions. While analogous result with the less restrictive stable Takens' theorem of \cite{eftekhari2018stabilizing} are beyond the scope of this paper, our initial work suggests that an additional assumption on the injectivity of the map $\pi_x$ in the nonlinear case. This matches the message that algorithms based on the closeness principle are sensitive to a variety of system-dependent properties that may be independent of causal structure, and supports the role of injectivity in methods such as \cite{cummins2015efficacy,harnack2017topological}.
    \item The necessity of a causal sufficiency assumption for Proposition \ref{thm:causality-stable-takens-iff} indicates that practitioners using methods based on the closeness principle should carefully assess whether results may be due to common counfounding, which may be indistinguishable from causal interaction to this class of algorithms.\footnote{Methods based on interventions (see, e.g., \cite{peters2020causal}) are able to draw conclusions about causal structure without this strong assumption, but unlike state-space methods, require the ability to perform experiments.}
\end{itemize}

\section{Typical systems don't satisfy the closeness principle}
\label{sec:experiments}

The results in Section \ref{sec:theory} provide insight into how causal inference techniques that are motivated by versions of the closeness principle might be justified. In practice, however, we find that many common coupled systems do not satisfy the closeness principle --- providing further evidence that the algorithms reviewed in Section \ref{sec:background-causal-inference} are sensitive to properties of the system (and its relationship to the measurement function) other than causal structure. In this section we empirically explore the preservation of geometric structure in common coupled systems and discuss implications for causal inference techniques based on stability. This stability perspective is instructive for understanding when and why closeness-principle-based causal inference methods may be well-justified.

We explore three common coupled nonlinear systems. In each, $x$ and $y$ are multidimensional, and $x$ drives $y$ with coupling strength $C$.
\begin{enumerate}
	\item The identical unidirectionally-coupled H\'enon maps,
	\begin{equation*}
		\begin{cases}
			x_1[t+1] = 1.4 - x_1^2[t] + 0.3 x_2[t] \\
			x_2[t+1] = x_1[t] \\
			y_1[t+1] = 1.4 - \left( C x_1[t] y_1[t] + (1-C) y_1^2[t] \right) + 0.3 y_2[t] \\
			y_2[t+1] = y_1[t],
		\end{cases}
	\end{equation*}
	with initial condition $(x_0,y_0) = [0.7, 0.0, 0.91, 0.7]^T$. To construct delay embeddings, we use the measurements
	\begin{equation*}
		\gamma_x = (x_1,x_2)[t] \mapsto x_1[t] ~\text{and}~ \gamma_y = (y_1,y_2)[t] \mapsto y_1[t]
	\end{equation*}
	so that
	\begin{equation*}
		\varphi_x = (x_1,x_2,y_1,y_2)[t] \mapsto x_1[t] ~\text{and}~ \varphi_y = (x_1,x_2,y_1,y_2) \mapsto y_1[t].
	\end{equation*}
	Identical synchronization (i.e., $x(t) = y(t)$) is induced just before $C = 0.7$ \cite{schiff1996detecting}. This system is used in, e.g., \cite{schiff1996detecting,quianquiroga2000learning,palus2007directionality,romano2007estimation,faes2008mutual,krakovska2018comparison}.
	\item A R\"ossler system coupled to a Lorenz system,
	\begin{equation*}
		\begin{cases}
			\dot{x}_1(t) = -6 \left( x_2(t) + x_3(t) \right) \\
			\dot{x}_2(t) = 6 \left( x_1(t) + 0.2 x_2(t) \right) \\
			\dot{x}_3(t) = 6 \left[ 0.2 + x_3(t) \left(x_1(t) - 5.7\right) \right] \\
			\dot{y}_1(t) = 10 \left( -y_1(t) + y_2(t) \right) \\
			\dot{y}_2(t) = 28 y_1(t) - y_2(t) - y_1(t) y_3(t) + C x_2^2(t) \\
			\dot{y}_3(t) = y_1(t) y_2(t) - \tfrac83 y_3(t),
		\end{cases}
	\end{equation*}
	with initial condition $(x_0,y_0) = [0.0, 0.0, 0.4, 0.3, 0.3, 0.3]^T$. The system trajectory was integrated with MATLAB's \texttt{ode45} and a sampling period of $dt=0.025$ was used to generate simulation data. To construct delay embeddings, we use the measurements
	\begin{equation*}
		\gamma_x = (x_1,x_2,x_3)(t) \mapsto x_2(t) ~\text{and}~ \gamma_y = (y_1,y_2,y_3)(t) \mapsto y_2(t)
	\end{equation*}
	so that
	\begin{equation*}
		\varphi_x = (x_1,x_2,x_3,y_1,y_2,y_3)(t) \mapsto x_2(t) ~\text{and}~ \varphi_y = (x_1,x_2,x_3,y_1,y_2,y_3)(t) \mapsto y_2(t).
	\end{equation*}
	Synchronization is induced just before $C = 3$ \cite{krakovska2018comparison}. This system is used in, e.g., \cite{levanquyen1999nonlinear,quianquiroga2000learning,andrzejak2003bivariate,palus2007directionality,krakovska2016testing,krakovska2018comparison,bahamonde2023usefulness}.
	\item The nonidentical unidirectionally coupled R\"ossler systems,
	\begin{equation*}
		\begin{cases}
			\dot{x}_1(t) = -\omega_1 x_2(t) - x_3(t) \\
			\dot{x}_2(t) = \omega_1 x_1(t) + 0.15 x_2(t) \\
			\dot{x}_3(t) = 0.2 + x_3(t) \left( x_1(t) - 10 \right) \\
			\dot{y}_1(t) = -\omega_2 y_2(t) - y_3(t) + C \left( x_1(t) - y_1(t) \right) \\
			\dot{y}_2(t) = \omega_2 y_1(t) + 0.15 y_2(t) \\
			\dot{y}_3(t) = 0.2 + y_3(t) \left( y_1(t) - 10 \right),
		\end{cases}
	\end{equation*}
	with $\omega_1 = 1.015$, $\omega_2 = 0.985$, and initial condition $(x_0,y_0) = [0.0, 0.0, 0.4, 0.0, 0.0, 0.4]^T$. To construct delay embeddings, we use the measurements
	\begin{equation*}
		\gamma_x = (x_1,x_2,x_3)(t) \mapsto x_3(t) ~\text{and}~ \gamma_y = (y_1,y_2,y_3)(t) \mapsto y_3(t)
	\end{equation*}
	so that
	\begin{equation*}
		\varphi_x = (x_1,x_2,x_3,y_1,y_2,y_3)(t) \mapsto x_3(t) ~\text{and}~ \varphi_y = (x_1,x_2,x_3,y_1,y_2,y_3)(t) \mapsto y_3(t).
	\end{equation*} Generalized synchronization \cite{rulkov1995generalized} is induced at approximately $C = 0.12$ \cite{krakovska2020implementation}. This system is used in, e.g., \cite{smirnov2005detection,palus2007directionality,krakovska2020implementation}; slight variants of this system are used in, e.g., \cite{rulkov1995generalized,hirata2010identifying,tajima2015untangling,hirata2016detecting,krakovska2018comparison,bahamonde2023usefulness}.
\end{enumerate}
We remove transients from the time series of each system by discarding the first 1000 generated samples.

\subsection{Characterizing the stability of common systems}
\label{sec:experiments-stability}

\begin{figure*}
	\centering
	\includegraphics[width=\textwidth]{\analysispath/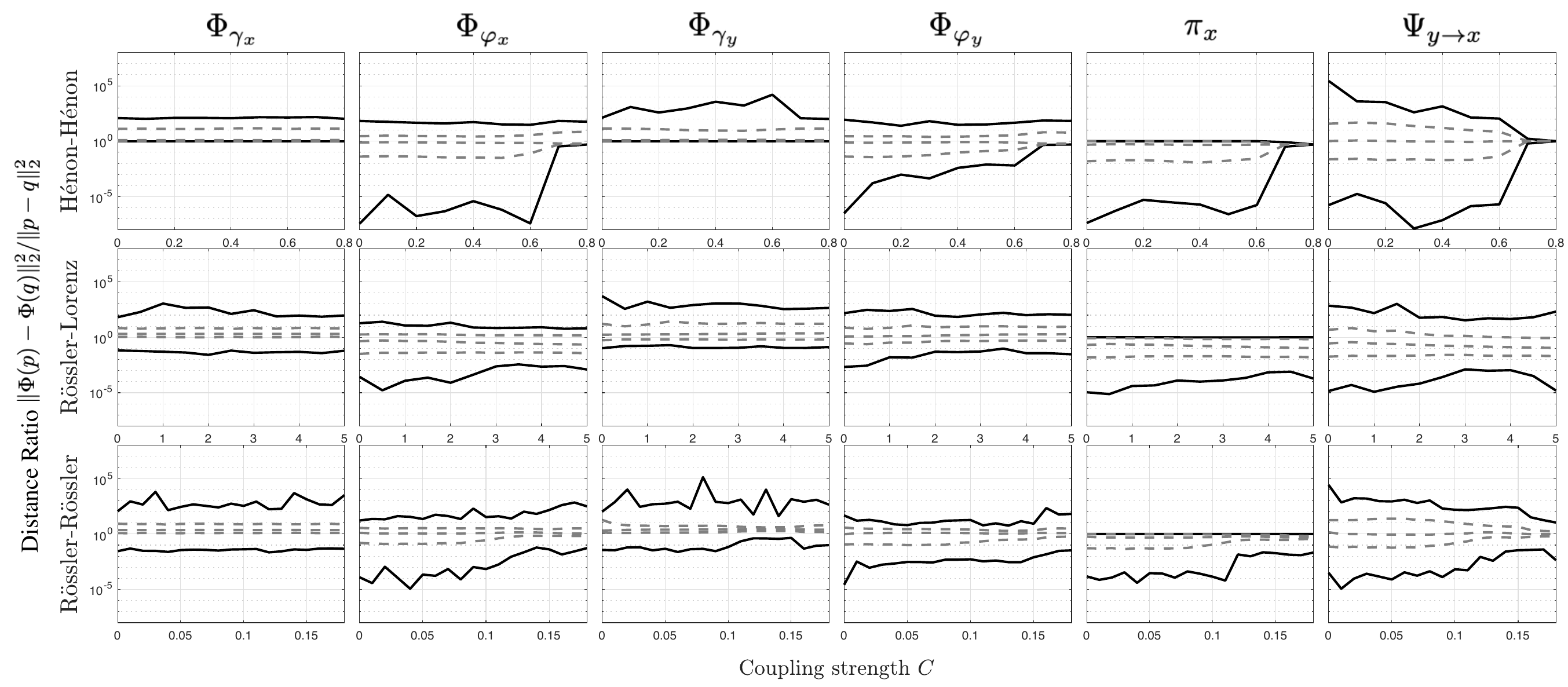}
	\caption{Distance ratios $\norm{\Phi(p)-\Phi(q)}_2^2 / \norm{p-q}_2^2$ of various maps $\Phi$ as coupling strength $C$ increases. Solid dark traces denote upper and lower isometry constants; lighter dashed traces denote 5\textsuperscript{th}, 50\textsuperscript{th}, and 95\textsuperscript{th} percentiles of distance ratios. Each statistic is computed from 5,000 sampled point pairs.}
	\label{fig:stability}
\end{figure*}

\begin{figure*}
	\centering
	\includegraphics[width=\textwidth]{\analysispath/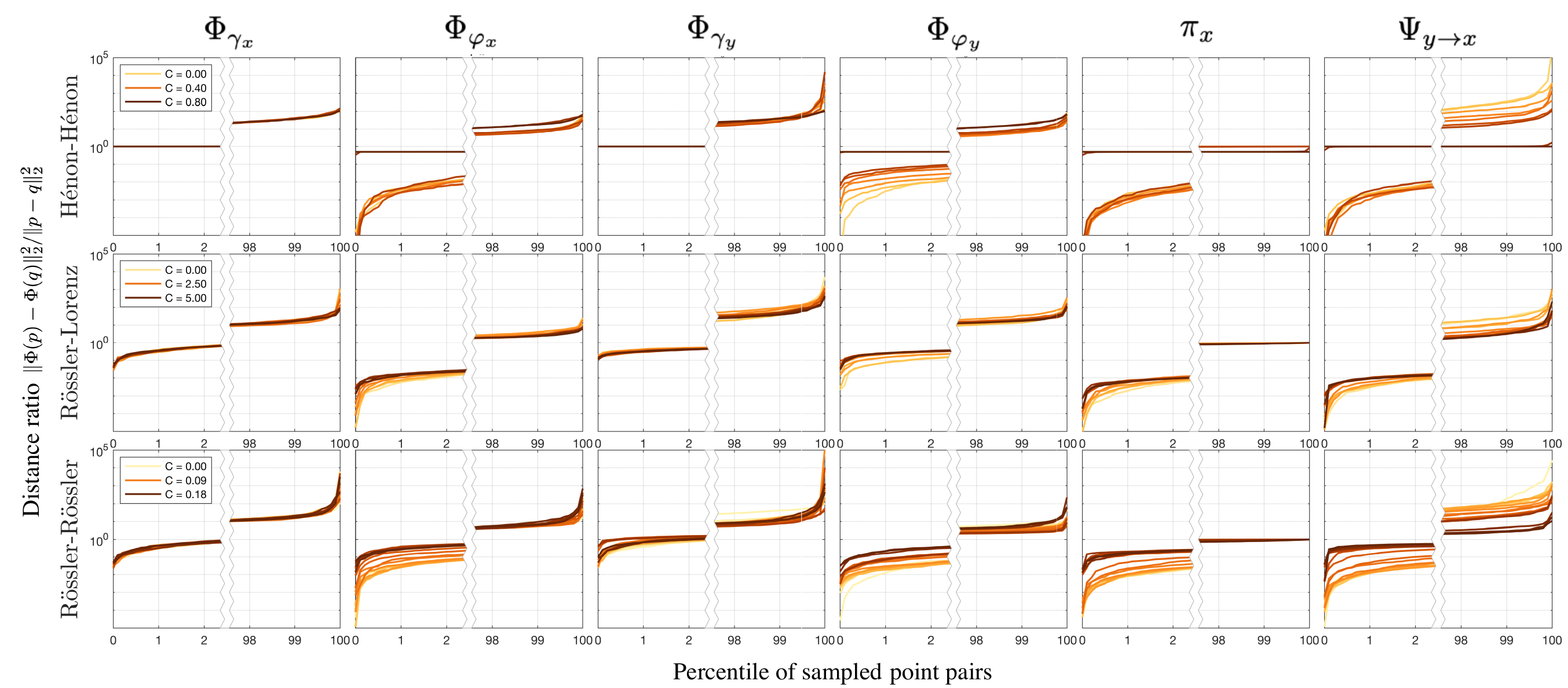}
	\caption{Percentiles of distance ratios $\norm{\Phi(p)-\Phi(q)}_2^2 / \norm{p-q}_2^2$ for various maps $\Phi$ as coupling strength $C$ increases. Darker colored traces represent stronger coupling. Tightly clustered traces indicate that map stability is less sensitive to $C$, while more broadly spread traces indicate that map stability is more sensitive to $C$.}
	\label{fig:stability-percentile}
\end{figure*}

We first empirically compute lower and upper isometry constants (i.e., the minimum and maximum amounts by which maps expand or contract distances between point pairs) for the relevant maps in Figure \ref{fig:mappings-all} for each of the nonlinear coupled systems described above. These isometry constants, shown in Figure \ref{fig:stability}, are each computed by randomly sampling 5,000 point pairs $(p,q)$ from the domain of each map. The empirical isometry constants (minimum and maximum distance ratios) represented by dark traces in Figure \ref{fig:stability} are used in the theory in Section \ref{sec:theory}.

Many closeness-principle-based methods operate not on isometry constants themselves, but on other statistics describing how certain maps tend to preserve distances between point pairs. To provide a more fulsome understanding of the distribution of the distance ratios of each mapping, the dashed lines in Figure \ref{fig:stability} show the 5\textsuperscript{th}, 50\textsuperscript{th}, and 95\textsuperscript{th} percentiles of the distance ratio computed for the 5,000 simulated point pairs. We also show, in Figure \ref{fig:stability-percentile}, how very small and large percentiles of these distributions of distance ratios are affected by coupling strength. Statistics such as these may be helpful in practice to reduce the sensitivity of purely distance-based tests to noise. Isometry increasing with coupling strength $C$ is indicated in Figure \ref{fig:stability-percentile} by traces moving closer to unity as $C$ as they become darker.

As expected because the $x$ subsystem is independent of the coupling strength $C$, the distance ratio of $\Phi_{\gamma_x} \colon M_x \to N_x$ remains constant (up to ``error'' stemming from the Monte-Carlo simulation used to compute distance ratios) as $C$ increases. The pattern of distance preservation not changing with coupling strength $C$ is reflected in Figure \ref{fig:stability} by roughly horizontal lines (indicating that distance preservation is static as $C$ increases along the horizontal axis), and in Figure \ref{fig:stability-percentile} by roughly overlapping lines (indicating that distance preservation is static as  $C$ increases with darker traces). Meanwhile, the stability of $\Phi_{\varphi_x}$, $\Phi_{\gamma_y}$, and $\Phi_{\varphi_y}$, each of which encode traits of the complete system, does in fact change as $C$ increases.

The map $\Phi_{\varphi_x} \colon M_{xy} \to N_x$ maps the entire coupled system $(x,y)$ to the delay embedding of only the \emph{driving} subsystem. As coupling strength $C$ increases (along the horizontal axis in Figure \ref{fig:stability} and as traces become darker in Figure \ref{fig:stability-percentile}), properties of the response system could change in ways that are not reflected in $N_x$. We thus expect $\Phi_{\varphi_x}$ to become less stable as $C$ begins to increase (allowing $\varphi_x$ to obtain less information about the complete $(x,y)$ system), and then to become more stable as $C$ increases further and synchronization is induced (allowing the state of the driving subsystem that $\varphi_x$ measures to again reveal the state of the entire coupled system). The magnitude of this effect, however, depends on an amalgamation of system- and measurement-specific traits. In the systems considered here, the second part of the pattern is indeed borne out in the H\'enon-H\'enon system, in which \emph{identical} synchronization occurs just before $C = 0.7$. In the remaining two systems, however, we see a general trend of \emph{increasing} stability as $C$ increases --- even before the onset of synchronization. In the cases where $\Phi_{\varphi_x}$ is unstable, it is usually because of an unfavorable \emph{lower} isometry constant, indicating that $\Phi_{\varphi_x}$ collapses points that are far on $M_{xy}$ to similar states of $N_x$. This pattern is particularly evident in Figure \ref{fig:stability-percentile}.

The map $\Phi_{\gamma_y} \colon M_y \to N_y$ transforms the response subsystem $y$ alone to its delay embedding. The effect of coupling strength $C$ on the stability of this map is to change the properties of the response subsystem. If this increased coupling strength results in a more complex, higher dimensional, or higher activity attractor for the response subsystem, we might expect the stability of the map $\Phi_{\gamma_y}$ to worsen. However, in Figures \ref{fig:stability} and \ref{fig:stability-percentile} we see limited effects of $C$ on stability, suggesting more optimistically that properties of the $y$ subsystem that are dependent on $C$ do not have significant impact on the stability of the delay embedding operator.

The map $\Phi_{\varphi_y} \colon M_{xy} \to N_y$ transforms the entire coupled system to the delay embedding of only the response subsystem. In this sense, $\Phi_{\varphi_y}$ is of particular interest in understanding the closeness principle because it quantifies --- using only information about stability --- the core principle that as coupling strength $C$ increases from zero the delay embedding of the \emph{response subsystem alone} ``contains information about'' the complete system (see Section \ref{sec:theory}). In the systems considered here, however, empirical evidence of this effect is modest. In cases where we see increasing isometry as $C$ increases (reflected in Figure \ref{fig:stability} by an increased convergence of traces towards unity and in Figure \ref{fig:stability-percentile} by traces converging toward unity as they become darker), it is not meaningfully different from those of $\Phi_{\varphi_x}$, suggesting that increasing (undirected) synchronization may be more at play than asymmetric coupling. In Figure \ref{fig:stability-percentile}, we see that the impact of $C$ is much different between the three systems. This suggests, yet again, that the strength of the ``closeness principle'' effect is impacted by system- and measurement-specific properties beyond strength of causal interaction alone.

The projection $\pi_x \colon M_{xy} \to M_x$ provides insight into the relative magnitudes of the complete system and driving subsystem. As expected, the stability of this map in the H\'enon-H\'enon system rapidly increases as synchronization occurs, with distance ratios converging to 0.5 as the two subsystems become identical. In the R\"ossler-R\"ossler system, $\pi_x$ becomes increasingly isometric as $C$ increases (particularly when comparing percentiles of distance ratios), suggesting stronger synchronization. As expected, the highest percentiles of distance ratios are all close to 1 (reflecting states in $(x,y)$ where $y$ has large magnitude while $x$ has low magnitude). However, the variance observed in lower percentiles of distance preservation suggest caution for the applicability of our theory in light of Assumption \ref{ass:inclusion-map-expansivity}.

Most relevant to analyzing closeness-principle-based methods is $\Psi_{y \to x} \colon N_y \to N_x$, which empirically describes the relationship that Conjecture \ref{conj:closeness-principle} is concerned with. Conjecture \ref{conj:closeness-principle} postulates that when $C = 0$ (i.e, $x \perp y$) $\Psi_{y \to x}$ will not be stable, while when $C > 0$ (i.e., $x \to y$) $\Psi_{y \to x}$ will be more stable. Figures \ref{fig:stability}--\ref{fig:stability-percentile} indeed reflect some of this pattern: as $C$ increases, $\Psi_{y \to x}$ becomes more stable. However, we do not consistently observe $\Psi_{y \to x}$ becoming \emph{less expansive} as $C$ increases, the key pattern shown in our carefully controlled linear test system of Section \ref{sec:theory-linear-phiy-example}. As we show next, this limitation is inherited by existing distance-preservation-based heuristic methods for detecting the direction of causal influence.

\subsection{Heuristics used by other methods}
\label{sec:experiments-heuristics}

\begin{figure*}
	\centering
	\includegraphics[width=\textwidth]{\analysispath/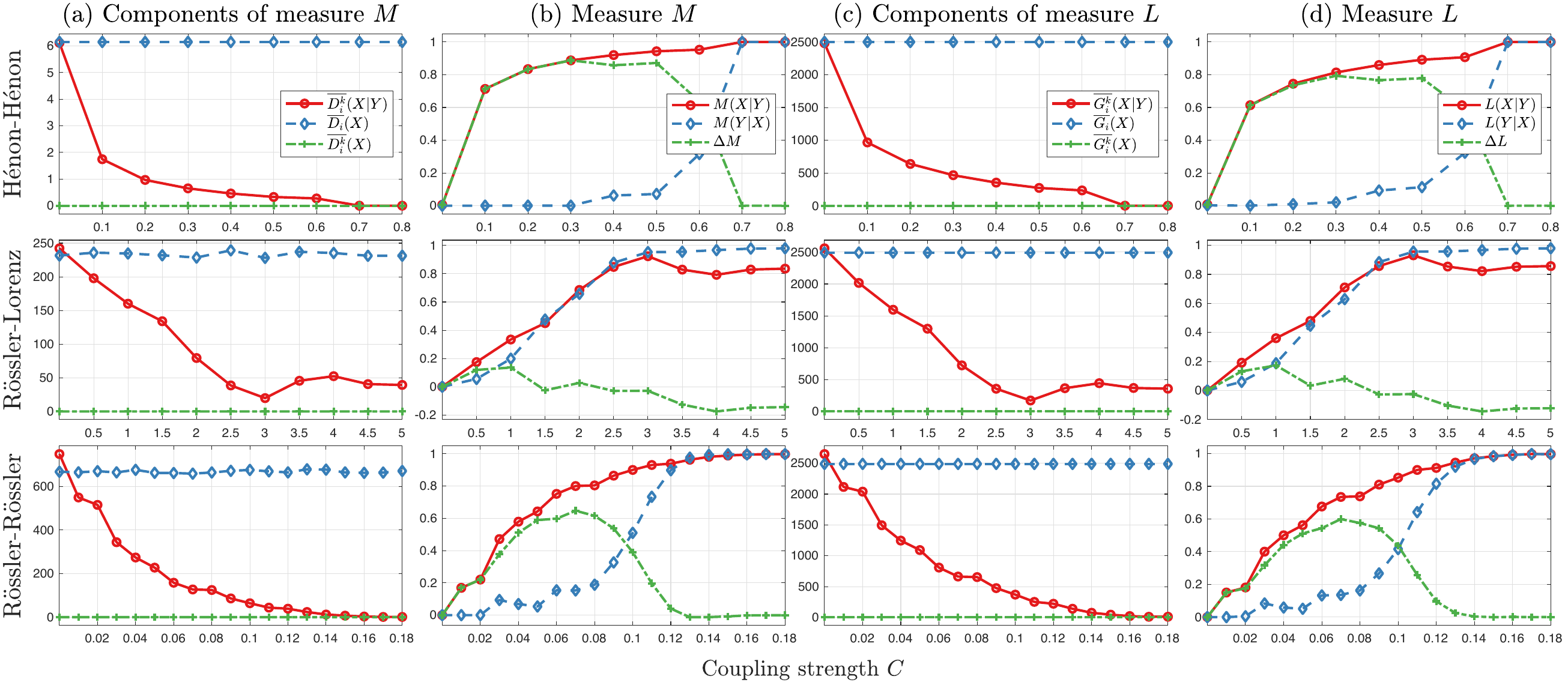}
	\caption{Other heuristics based on the distance preservation of $\Psi_{y \to x}$ for the coupled nonlinear systems defined above as the system coupling strength $C$ increases. (a--b) directional influence metrics \eqref{eq:metric-andrzejak2003bivariate} of \cite{andrzejak2003bivariate} and means of constituent terms; (c--d) rank-based directional influence metric \eqref{eq:metric-chicharro2009reliable} of \cite{chicharro2009reliable} and means of constituent terms.}
	\label{fig:heuristics}
\end{figure*}

\begin{figure*}
	\centering
	\includegraphics[width=\textwidth]{\analysispath/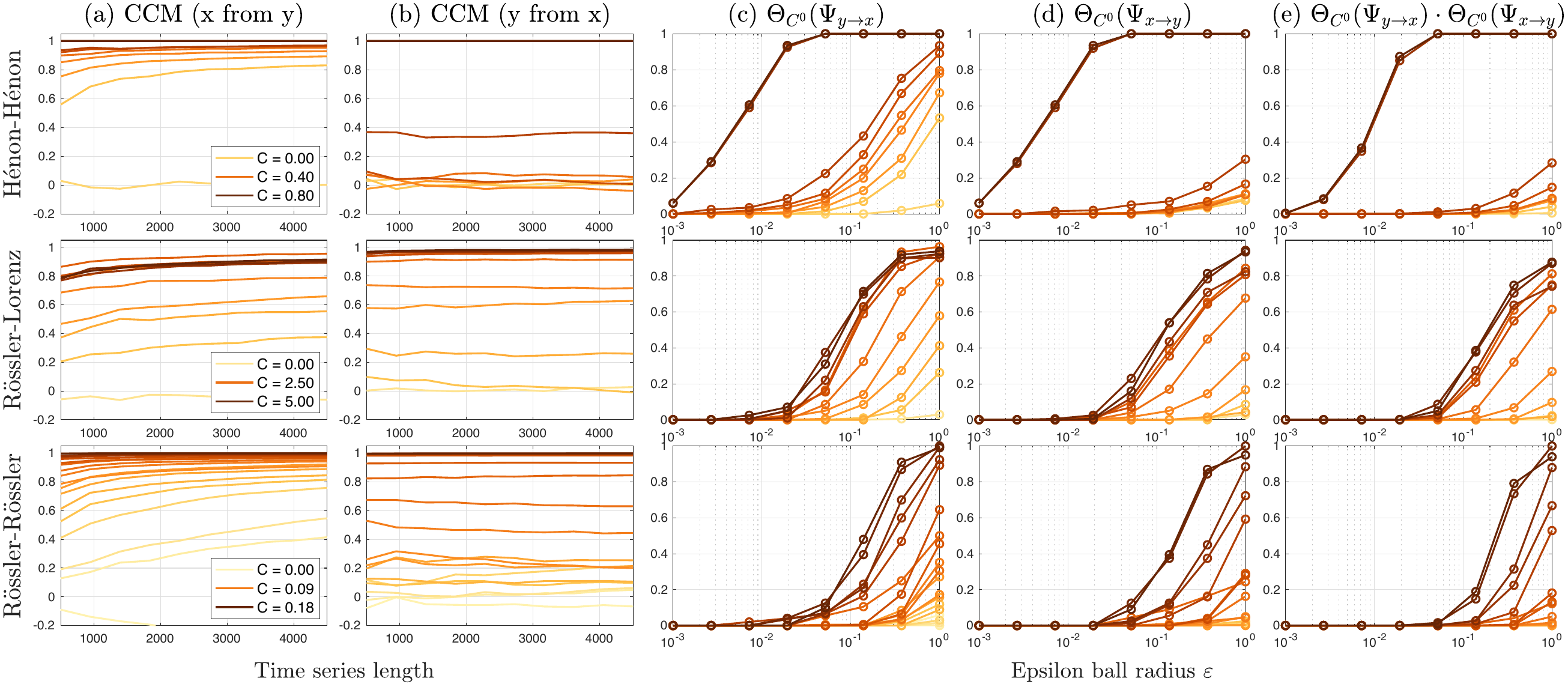}
	\caption{Other heuristics based on continuity-like properties of $\Psi_{y \to x}$ for the coupled nonlinear systems defined above as the system coupling strength $C$ increases. (a--b) convergent cross-mapping heuristic of \cite{sugihara2012detecting}; (c--e) continuity heuristic of \cite{pecora1995statistics} for (c) continuity, (d) inverse continuity, and (e) injectivity.}
	\label{fig:heuristics2}
\end{figure*}

Although Section \ref{sec:experiments-stability} found limited evidence for the existence of a purely-distance-based form of Conjecture \ref{conj:closeness-principle}, many techniques have obtained successful --- if sometimes inconsistent --- results using this principle. In this section we implement a sample of these heuristics and interpret their results in the context of our work.

Shown in Figure \ref{fig:heuristics} are results from the relative-distance-based metrics of \cite{andrzejak2003bivariate} and \cite{chicharro2009reliable} described in Section \ref{sec:background-causal-inference}. Both techniques, roughly speaking, compare how well nearest neighbors on $N_y$ correspond to nearest neighbors on $N_x$ [see \eqref{eq:metric-andrzejak2003bivariate} and \eqref{eq:metric-chicharro2009reliable}]; \cite{andrzejak2003bivariate} directly uses distances while \cite{chicharro2009reliable} uses ranked distances.

We see that both techniques are effective in detecting coupling between $x$ and $y$: as $C$ increases, the average distance to $k$ mutual neighbors ($\overline{D_i^k}(X \mid Y)$) and $\overline{G_i^k}(X \mid Y)$ smoothly decreases from the mean inter-point distance on $N_x$ ($\overline{D_i}(X)$ and $\overline{G_i}(X)$) to the mean distance to each point on $N_x$'s $k$ nearest neighbors ($\overline{D_i^k}(X)$ and $\overline{G_i^k}(X)$). For both the H\'enon-H\'enon and R\"ossler-R\"ossler systems, these methods effectively detect coupling direction: $M(X \mid Y)$ and $L(X \mid Y)$ rise quickly with $C$, while $M(Y \mid X)$ and $L(Y \mid X)$ --- the metrics used to quantify coupling from $y \to x$ --- begin to increase only at the onset of synchronization. For the R\"ossler-Lorenz system, however, the metrics quantifying $x \to y$ and $y \to x$ rise simultaneously, and it is not possible to detect coupling direction. Previous work \cite{arnhold1999robust,quianquiroga2000learning,chicharro2009reliable} has showed how these measures depend on system-specific properties like attractor dimension, so it is unsurprising that we see large differences in efficacy between systems.

As discussed in the introduction, Takens' theorem guarantees only that $\Psi_{y \to x}$ is continuous. We next consider two heuristic methods that are more closely tied to the \emph{continuity} of $\Psi_{y \to x}$ than its stability.

The convergent cross-mapping (CCM) algorithm of \cite{sugihara2012detecting} (see Section \ref{sec:background-causal-inference}) relies on the principle that when $x \to y$, $N_y$ will become increasingly predictive of the measurements obtained from $x$ as time series length increases. As shown in Figure \ref{fig:heuristics2}(a--b), we find that this method is generally effective in the three systems we consider. With the exception of the R\"ossler-R\"ossler example, $N_y$ is unable to reconstruct $x$ when $C = 0$, and cross-map efficacy increases monotonically with $C$. Meanwhile, as shown in Figure \ref{fig:heuristics2}(b), the ability of $x$ to reconstruct $N_y$ does not increase with time series length, and high reconstruction accuracy is achieved only at values of $C$ close to the onset of synchronization.

An alternate heuristic for continuity developed by \cite{pecora1995statistics} was proposed as part of a rigorous theoretical framework detecting asymmetric coupling relationships by \cite{cummins2015efficacy} (see Section \ref{sec:introduction}). This heuristic finds sets on $N_x$ and $N_y$ that satisfy the epsilon-delta definition of continuity, then assesses the probability that these points were found by chance using a specific null hypothesis \cite{pecora1995statistics}. Evidence for continuity, denoted by $\Theta_{C^0}(\Psi_{y \to x})$, is provided by this probability increasing to near-unity as $\varepsilon$ increases. Evidence for continuity of the inverse map $\Psi_{x \to y}$ is computed similarly. Because the continuity of a map implies that it is one-to-one, evidence that a map is homeomorphic is provided by the combination of both continuity and inverse continuity. Thus, the product $\Theta_{C^0}(\Psi_{y \to x}) \cdot \Theta_{C^0}(\Psi_{x \to y})$ increasing to unity with $\varepsilon$ provides evidence that the map is homeomorphic. As shown in Figure \ref{fig:heuristics2}(c--e), we find evidence in all three systems of increasing continuity as $C$ increases. In the H\'enon-H\'enon map, $\Theta_{C^0}(\Psi_{y \to x})$ increases with $C$ much quicker than $\Theta_{C^0}(\Psi_{x \to y})$ does (until $C$ becomes large enough for identical synchronization), providing evidence for the relationship $x \to y$. In the R\"ossler-Lorenz and R\"ossler-R\"ossler systems we similarly see $\Theta_{C^0}(\Psi_{y \to x})$ and $\Theta_{C^0}(\Psi_{x \to y})$ increase with $C$. However, here we see less asymmetry between $\Theta_{C^0}(\Psi_{y \to x})$ and $\Theta_{C^0}(\Psi_{x \to y})$, providing less value for determining the \emph{direction} of asymmetric coupling. This closely follows the results from CCM shown in Figure \ref{fig:heuristics2}(a--b).

\section{Discussion}
\label{sec:discussion}

This paper examined the ``closeness principle'' that underlies many state-space methods for detecting causal interaction in coupled dynamical systems --- that is, the principle that when $x \to y$, pairs of points that are close together on the delay embedding of $y$ will correspond to points that are also close on the delay embedding of $x$. We applied the recent stable Takens' theorems of \cite{yap2011stable,eftekhari2018stabilizing} to develop guarantees for one straightforward method for operationalizing the closeness principle to detect causal interaction.

\emph{Theoretical takeaways for closeness-principle-based methods:} While our results provide a foundation for developing more theoretically-grounded tests for causal interaction, they also suggest that the distance preservation of delay embeddings is strongly sensitive to a medley of system properties that are unrelated to causal structure. In addition, we showed empirically that common coupled dynamical systems do not satisfy literal definitions of the closeness principle, and that typical closeness-principle-based heuristics are dependent on not only the direction of asymmetric coupling but also on a myriad of other system-specific traits.

Continuity-based properties implied by the closeness principle are used to justify many state-space methods in the literature. A less explored fundamental principle operationalized by some other heuristics is the \emph{injectivity} of maps between delay embeddings. Many approaches (e.g., \cite{harnack2017topological}) are justified on the basis that the causal structure underlying a coupled system often creates ``folds'' in the joint attractor $M_{xy}$ \cite{amigo2018detecting}. When this is the case, the (non)injectivity of $\Psi_{y \to x}$ or $\pi_x$ can provide evidence of causal structure. Although our theoretical results for the map $\Phi_{\varphi_y}$ in Section \ref{sec:theory-linear-phiy} focused on linear observation functions and systems, we found that a quantity measuring a graded notion of $\pi_x$'s noninjectivity plays a key role in developing a nonlinear analogue of Proposition \ref{thm:linear-forced-stable-takens} using the more general stable Takens' theorem of \cite{eftekhari2018stabilizing}. This is a more nuanced analogue to the assumption on the noninjectivity of $\pi_x$ in \cite{cummins2015efficacy}. The injectivity and invertibility of $\pi_x$ and $\Psi_{y \to x}$ are also closely related to the synchronization of the response system to the driving system.

\emph{Takeaways for practitioners:} Practitioners using state-space algorithms for detecting causal interactions should be aware that individual methods operationalize slightly different proxies for the continuity and injectivity of maps between delay embeddings, each of which is sensitive not only to changes in coupling strength, but also to a large collection of unmeasurable system-specific properties whose effects can easily be conflated coupling. Like \cite{hirata2016detecting,cliff2022unifying}, we suggest that practitioners apply a collection of these techniques rather than relying on the output of a single algorithm in this class of methods which may be overly-sensitive to specific system properties outside of coupling strength. The consistent output of multiple methods can be interpreted as circumstantial evidence for asymmetric coupling.

Previous work \cite{quianquiroga2000learning,smirnov2005detection,cobey2016limits,hirata2016detecting,krakovska2018comparison,munch2019frequently} has studied many of these confounding system properties and suggested that the success of closeness-principle-based methods depends on a smorgasboard of system properties not directly related to causal structure, including observation and process noise statistics \cite{chicharro2009reliable,schiecke2015convergent,cobey2016limits,jiang2016directed,amigo2018detecting}, relative system dimension and activity \cite{arnhold1999robust,quianquiroga2000learning,romano2007estimation,benko2018exact,krakovska2019correlation}, the injectivity of the combined system \footnote{Specifically, if the projection $\pi_x \colon M_{xy} \to M_x$ is injective, it and $\Psi_{y \to x}$ may be invertible. The noninjectivity of $\pi_x$ is often described as creating ``folds'' in $M_{xy}$ \cite{harnack2017topological,amigo2018detecting} that allow inference; see \cite[pg.~344]{cummins2015efficacy}.} \cite{levanquyen1999nonlinear,harnack2017topological,amigo2018detecting}, the presence of transient dynamics \cite{cobey2016limits} or generalized synchronization \cite{rulkov1995generalized,cobey2016limits,martin2020comment}, the dimension $m$ and time lag $\tau$ used to construct delay embeddings \cite{tajima2015untangling,schiecke2015convergent,krakovska2016testing,cobey2016limits}, and consistency with the (unverifiable) assumptions of Takens' theorem \cite{munch2019frequently}. Our results underline the strong and less explored role of the measurement function used to construct the delay embedding --- in particular, how well the measurement function aligns with structure of a system's attractor.

It is also important for practitioners to carefully consider the sense in which closeness-principle-based algorithms detect ``causal'' interactions. Precisely speaking, the causal relationship $x \to y$ implies that an intervention on $x$ would produce a corresponding change in $y$. The class of methods described in Section \ref{sec:background-causal-inference} provide only necessary conditions for the existence of causal relationships, so while it is justified to interpret a negative result from these tests as implying the \emph{absence} of a causal relationship, a positive result does not prove the \emph{existence} of a causal relationship. This interpretation is consistent with discussion of state-space methods in \cite{sugihara2012detecting,cummins2015efficacy,krakovska2018comparison,butler2023causal} --- and best practice for interpreting the output of other methods such as transfer entropy \cite{schreiber2000measuring} and Granger (predictive) causality \cite{granger1969investigating}. In our theoretical results, we required a strong causal sufficiency assumption --- which eliminates the possibility of, for instance, confounding interactions --- to derive necessary and sufficient conditions on the relationships between state space maps for the existence of causal interactions. Methods that do not implicitly or explicitly make this type of assumption instead require data with fortuitous invariance characteristics \cite{pfister2019learning} or the ability to perform (often difficult or impossible) experiments \cite{peters2020causal}.

\emph{Future work:} Our results suggest several avenues for future work. The theoretical results in Section \ref{sec:theory} were phrased in terms of simple isometry constants (i.e., minimum and maximum distance ratios); results that consider more nuanced statistics of how mappings between delay embeddings preserve distances may be more useful for developing practical and robust tests. Second, translating results phrased in terms of \emph{distance} preservation to consider \emph{nearest neighbor} preservation would help close the gap between the theory we have developed here and many algorithms used in practice. Although we expect theoretical guarantees on when nearest neighbor relationships are preserved to be more difficult to derive than purely distance-based guarantees, computing distance ratios for pairs of points is computationally expensive and sensitive to observation noise; nearest-neighbor-based tests have been used in practice both to determine embedding parameters (e.g., \cite{kennel1992determining,abarbanel1993local}) and for detecting causal interaction \cite{arnhold1999robust,quianquiroga2000learning,andrzejak2003bivariate}. Third, our results used the stable Takens' theorems of \cite{yap2011stable,eftekhari2018stabilizing}, which provide only \emph{sufficient} conditions for the preservation of distance ratios between a system and its delay embedding. A key finding of our work was that intrinsic system properties, and the relationship between a system and the measurement function used to observe it, can have large impacts on closeness-principle-based heuristics. The development of \emph{necessary} conditions for embedding stability would help to disentangle the impact of causal interaction from other system properties. Finally, while our results deal with coupled systems comprised of only two subsystems, the manifold filtration framework of \cite{cummins2015efficacy} might be used to extend these results to more complex systems.

\section*{Acknowledgements}

This work was supported by NSF grants CCF-1409422 and CCF-1350954, NIH grant 1R01NS115327, and the National Defense Science \& Engineering Graduate (NDSEG) Fellowship. M.R.O.\ thanks Adam Willats for helpful discussions.

\printbibliography

\clearpage
\appendix

\section{Proof of Proposition \ref{thm:linear-forced-stable-takens}}
\label{sup:proofs-linear-forced-stable-takens}

Proposition \ref{thm:linear-forced-stable-takens} shows that delay embeddings of systems in class $\mathcal{A}(d)$ constructed with linear measurement functions $\varphi_y$ that ``use information about $y$'' can be stable mappings only if $x \to y$.

\begin{proof}
	First, note that in the system matrix for a linear system
	\begin{equation}
		A = \begin{bmatrix} A_{xx} & A_{xy} \\ A_{yx} & A_{yy} \end{bmatrix},
		\label{eq:system-matrix-partitioned}
	\end{equation}
	$A_{yx} = 0$ when $x \not\to y$ but $A_{yx} \neq 0$ when $x \to y$. The proof will use this fact to show that the delay embedding operator $\Phi_{\varphi_y}$ is rank-deficient when $x \not\to y$ but full rank (in general) when $x \to y$.
	
	By assumption both the measurement operator $\varphi_y(\cdot) = \left< h_{\varphi_y}, \cdot \right>$ and the system \eqref{eq:system-matrix-partitioned} are linear. This allows the delay embedding operator $\Phi_{\varphi_y}$ to be written in matrix notation as the linear operator
	\begin{equation}
		\Phi_{\varphi_y} = \begin{bmatrix} h_{\varphi_y}^T e^{-0 \cdot A \tau} \\ h_{\varphi_y}^T e^{-1 \cdot A \tau} \\ \vdots \\ h_{\varphi_y}^T e^{-(m-1) A \tau} \end{bmatrix} \in \R^{m \times (n_x + n_y)}. \label{eq:linear-forced-stable-takens-matrix}
	\end{equation}
	The form of $e^{-kA\tau}$ --- and therefore the properties of $\Phi_{\varphi_y}$ --- depends on whether $x \to y$ or $x \not\to y$.
	
	\emph{Part 1: when $x \not\to y$, $\Phi_{\varphi_y}$ is not a stable map.} When $x \not\to y$, the system matrix \eqref{eq:system-matrix-partitioned} has $A_{yx} = 0$, so
	\begin{equation*}
		e^{-k A \tau} = \begin{bmatrix} e^{-k A_{xx} \tau} & \widetilde{A}_{xy,k} \\ 0_{n_y \times n_x} & e^{-k A_{yy} \tau} \end{bmatrix}
	\end{equation*}
	where $\widetilde{A}_{xy}^{(k)} = \sum_{j=0}^{\infty} (j!)^{-1} (-k \tau)^j \sum_{i=1}^j A_{xx}^{j-i} A_{xy} A_{yy}^{i-1}$. The matrix representation of $\Phi_{\varphi_y}$ is thus
	\begin{align*}
		\Phi_{\varphi_y} &= \begin{bmatrix} h_{\varphi_y}^T \begin{bmatrix} I_{n_x \times n_x} & 0_{n_x \times n_y} \\ 0_{n_y \times n_x} & I_{n_y \times n_y} \end{bmatrix} \\ h_{\varphi_y}^T \begin{bmatrix} e^{- A_{xx} \tau} & \widetilde{A}_{xy}^{(1)} \\ 0_{n_y \times n_x} & e^{- A_{yy} \tau} \end{bmatrix} \\ \vdots \\ h_{\varphi_y}^T \begin{bmatrix} e^{- (m-1) A_{xx} \tau} & \widetilde{A}_{xy}^{(m-1)} \\ 0_{n_y \times n_x} & e^{- (m-1) A_{yy} \tau} \end{bmatrix} \end{bmatrix} \\
		\intertext{Because $\varphi_y$ is linear and depends only on $y$, $h_{\varphi_y} = \begin{bmatrix} 0_{n_x}^T & h_{\gamma_y}^T \end{bmatrix}^T$:}
		&= \begin{bmatrix} \begin{bmatrix} 0_{n_x}^T & h_{\gamma_y}^T \end{bmatrix} \begin{bmatrix} I_{n_x \times n_x} & 0_{n_x \times n_y} \\ 0_{n_y \times n_x} & I_{n_y \times n_y} \end{bmatrix} \\ \begin{bmatrix} 0_{n_x}^T & h_{\gamma_y}^T \end{bmatrix} \begin{bmatrix} e^{- A_{xx} \tau} & \widetilde{A}_{xy}^{(1)} \\ 0_{n_y \times n_x} & e^{- A_{yy} \tau} \end{bmatrix} \\ \vdots \\ \begin{bmatrix} 0_{n_x}^T & h_{\gamma_y}^T \end{bmatrix} \begin{bmatrix} e^{- (m-1) A_{xx} \tau} & \widetilde{A}_{xy}^{(m-1)} \\ 0_{n_y \times n_x} & e^{- (m-1) A_{yy} \tau} \end{bmatrix} \end{bmatrix} \\
		&= \begin{bmatrix} & h_{\gamma_y}^T e^{0 \cdot A_{yy} \tau} \\ 0_{m \times n_x} & \vdots \\ & h_{\gamma_y}^T e^{-(m-1) \cdot A_{yy} \tau} \end{bmatrix},
	\end{align*}
	which is rank-deficient. The matrix representation of $\Phi_{\varphi_y}$ therefore has nontrivial nullspace, so there exist distinct $p, q$ such that $\Phi_{\varphi_y} (p - q) = 0$. Thus, there exists $p \neq q$ such that
	\begin{equation*}
		\frac{\norm{\Phi_{\varphi_y}(p)-\Phi_{\varphi_y}(q)}_2^2}{\norm{p-q}_2^2} = 0,
	\end{equation*}
	so $\Phi_{\varphi_y}$ is not stable.
	
	\emph{Part 2: when $x \to y$, $\Phi_{\varphi_y}$ may be a stable map.} When the linear system has underlying causal structure $x \to y$, the system matrix \eqref{eq:system-matrix-partitioned} satisfies $A_{yx} \neq 0$. We then have
	\begin{equation*}
		e^{-kA\tau} = \begin{bmatrix} \widetilde{A}_{xx}^{(k)} & \widetilde{A}_{xy}^{(k)} \\ \widetilde{A}_{yx}^{(k)} & \widetilde{A}_{yy}^{(k)} \end{bmatrix}.
	\end{equation*}
	The blocks $\widetilde{A}_{xx}^{(k)}$, $\widetilde{A}_{xy}^{(k)}$, $\widetilde{A}_{yx}^{(k)}$, and $\widetilde{A}_{yy}^{(k)}$ do not admit simple expressions --- but they are in general nonzero. The matrix representation of $\Phi_{\varphi_{y}}$ is thus
	\begin{align*}
		\Phi_{\varphi_y} &= \begin{bmatrix} h_{\varphi_y}^T \begin{bmatrix} I_{n_x \times n_x} & 0_{n_x \times n_y} \\ 0_{n_y \times n_x} & I_{n_y \times n_y} \end{bmatrix} \\ h_{\varphi_y}^T \begin{bmatrix} \widetilde{A}_{xx}^{(1)} & \widetilde{A}_{xy}^{(1)} \\ \widetilde{A}_{yx}^{(1)} & \widetilde{A}_{yy}^{(1)} \end{bmatrix} \\ \vdots \\ h_{\varphi_y}^T \begin{bmatrix} \widetilde{A}_{xx}^{(m-1)} & \widetilde{A}_{xy}^{(m-1)} \\ \widetilde{A}_{yx}^{(m-1)} & \widetilde{A}_{yy}^{(m-1)} \end{bmatrix} \end{bmatrix} \\
		\intertext{Because $\varphi_y$ is linear and depends only on $y$, $h_{\varphi_y} = \begin{bmatrix} 0_{n_x}^T & h_{\gamma_y}^T \end{bmatrix}^T$:}
		&= \begin{bmatrix} \begin{bmatrix} 0_{n_x}^T & h_{\gamma_y}^T \end{bmatrix} \begin{bmatrix} I_{n_x \times n_x} & 0_{n_x \times n_y} \\ 0_{n_y \times n_x} & I_{n_y \times n_y} \end{bmatrix} \\ \begin{bmatrix} 0_{n_x}^T & h_{\gamma_y}^T \end{bmatrix} \begin{bmatrix} \widetilde{A}_{xx}^{(1)} & \widetilde{A}_{xy}^{(1)} \\ \widetilde{A}_{yx}^{(1)} & \widetilde{A}_{yy}^{(1)} \end{bmatrix} \\ \vdots \\ \begin{bmatrix} 0_{n_x}^T & h_{\gamma_y}^T \end{bmatrix} \begin{bmatrix} \widetilde{A}_{xx}^{(m-1)} & \widetilde{A}_{xy}^{(m-1)} \\ \widetilde{A}_{yx}^{(m-1)} & \widetilde{A}_{yy}^{(m-1)} \end{bmatrix} \end{bmatrix} \\
		&= \begin{bmatrix} 0_{1 \times n_x} & h_{\gamma_y}^T \\ h_{\gamma_y}^T \widetilde{A}_{yx}^{(1)} & h_{\gamma_y}^T \widetilde{A}_{yy}^{(1)} \\ \vdots \\ h_{\gamma_y}^T \widetilde{A}_{yx}^{(m-1)} & h_{\gamma_y}^T \widetilde{A}_{yy}^{(m-1)} \end{bmatrix},
	\end{align*}
	which in general has full rank. If the conditions of Theorem \ref{thm:linear-stable-takens} are satisfied, this matrix will have favorable structure, so $\Phi_{\varphi_y}$ can satisfy \eqref{eq:stable-embedding}. Note that $\widetilde{A}_{yx}$ and $\widetilde{A}_{yy}$ in general depend on all blocks of $A$, so the relevant properties of $\Phi_{\varphi_y}$ also depend on the behavior of the $x$ subsystem.
\end{proof}

\section{Proof of Proposition \ref{thm:causality-stable-takens-1}}
\label{sup:proofs-causality-stable-takens-1}

Proposition \ref{thm:causality-stable-takens-1} relates distance preservation in $\{ \Phi_{\gamma_x}, \Phi_{\gamma_y}, \Phi_{\varphi_x}, \Phi_{\varphi_y} \}$ to lower or upper bounds (depending on the causal structure) on the distance preservation of $\Phi_{y \to x}$.

\begin{proof}
	By Takens' theorem, $M_x$ and $N_x$ are diffeomorphic, and $M_y$ and $N_y$ are diffeomorphic, while the relationship of $M_{xy}$ to $N_x$ and $N_y$ depends on the measurement operator and underlying causal structure of the system.
	
	First note that (\ref{eq:stabletakens-Phi-gammax}--\ref{eq:stabletakens-Phi-phiy}) imply that
	\begin{align}
		\frac{1}{u_{\gamma_x}} &\leq \frac{\norm{\Phi^{-1}_{\gamma_x}(p)-\Phi^{-1}_{\gamma_x}(q)}_2^2}{\norm{p-q}_2^2} \leq \frac{1}{l_{\gamma_x}} \quad\forall~ p,q \in N_x, \label{eq:stabletakens-Phiinv-gammax} \\
		\frac{1}{u_{\gamma_y}} &\leq \frac{\norm{\Phi^{-1}_{\gamma_y}(p)-\Phi^{-1}_{\gamma_y}(q)}_2^2}{\norm{p-q}_2^2} \leq \frac{1}{l_{\gamma_y}} \quad\forall~ p,q \in N_y. \label{eq:stabletakens-Phiinv-gammay} \\
		\frac{1}{u_{\varphi_x}} &\leq \frac{\norm{\Phi^{-1}_{\varphi_x}(p)-\Phi^{-1}_{\varphi_x}(q)}_2^2}{\norm{p-q}_2^2} \leq \frac{1}{l_{\varphi_x}} \quad\forall~ p,q \in N_x, \label{eq:stabletakens-Phiinv-phix} \\
		\frac{1}{u_{\varphi_y}} &\leq \frac{\norm{\Phi^{-1}_{\varphi_y}(p)-\Phi^{-1}_{\varphi_y}(q)}_2^2}{\norm{p-q}_2^2} \leq \frac{1}{l_{\varphi_y}} \quad\forall~ p,q \in N_y. \label{eq:stabletakens-Phiinv-phiy}
	\end{align}
	(To see this, let $p = \Phi_{\varphi_z}^{-1}(p')$ and $q = \Phi_{\varphi_z}^{-1}(q')$, $p', q' \in N_z$ and note that $\Phi_{\varphi_z}$ is bijective by Takens' theorem, so $\Phi_{\varphi_z}(\Phi_{\varphi_z}^{-1}(\cdot))$ is the identity map.)
	
	These near-isometries are the first key ingredient of our result about $\Psi_{y \to x} \colon N_y \to N_x$; they allow us to relate distances on the attractors $\{M_{xy}, M_x, M_y\}$ to distances on their delay embeddings.
	
	The second key ingredient allows us to relate distances on the attractors of different subsystems. The mapping $\pi_x \colon M_{xy} \to M_x$ is linear and nonexpansive,
	\begin{equation}
		0 \leq \frac{\norm{\pi_x(p)-\pi_x(q)}_2^2}{\norm{p-q}_2^2} = \frac{\norm{\pi_x(p-q)}_2^2}{\norm{p-q}_2^2} \leq 1, \label{eq:projection-nonexpansive}
	\end{equation}
	and similarly the projection $\pi_y \colon M_{xy} \to M_y$ is linear and nonexpansive. By contrast, the inclusion map $\iota_x \colon M_y \to M_{xy}$ is noncontractive: for $p_y, q_y \in M_y$ and contemporaneous $p_x, q_x \in M_x$,
	\begin{equation}
		\frac{\norm{\iota_x(p_y)-\iota_x(q_y)}_2^2}{\norm{p_y-q_y}_2^2} = \frac{\norm{\begin{bmatrix} p_x \\ p_y \end{bmatrix} - \begin{bmatrix} q_x \\ q_y \end{bmatrix}}_2^2}{\norm{p_y-q_y}_2^2} = \frac{\norm{p_x-q_x}_2^2+\norm{p_y-q_y}_2^2}{\norm{p_y-q_y}_2^2} \geq 1,
		\label{eq:inclusion-map-noncontractive}
	\end{equation}
	and similarly the inclusion map $\iota_y \colon M_x \to M_{xy}$ is noncontractive. We now apply these facts to bound the stability of $\Psi_{y \to x}$ based on whether $x \to y$ or $y \to x$.
	
	\begin{figure}[t]
		\centering
		\begin{tabular}{cc}
			$\boxed{\text{(a)}~~x \to y\vphantom{\perp}}$ & $\boxed{\text{(b)}~~y \to x\vphantom{\perp}}$ \\[0.2cm]
			\begin{tikzpicture}
				\matrix (m) [matrix of math nodes, row sep=3em, column sep=5em, minimum width=4em, minimum height=3em, nodes={draw=black, minimum height=2.5em}] { M_x & M_{xy} \\ N_x & N_y \\ };
				\path[-{Latex[length=0.7em]}]
				(m-1-1) edge node [left] {$\Phi_{\gamma_x}$} (m-2-1)
				(m-2-2) edge node [right] {$\Phi_{\varphi_y}^{-1}$} (m-1-2)
				(m-1-2) edge node [above] {$\pi_x$} (m-1-1);
			\end{tikzpicture} &
			\begin{tikzpicture}
				\matrix (m) [matrix of math nodes, row sep=3em, column sep=5em, minimum width=4em, minimum height=3em, nodes={draw=black, minimum height=2.5em}] { M_{xy} & M_y \\ N_x & N_y \\ };
				\path[-{Latex[length=0.7em]}]
				(m-1-1) edge node [left] {$\Phi_{\varphi_x}$} (m-2-1)
				(m-2-2) edge node [right] {$\Phi_{\gamma_y}^{-1}$} (m-1-2)
				(m-1-2) edge node [above] {$\iota_x$} (m-1-1);
			\end{tikzpicture}
		\end{tabular}
		\caption{Possible intermediate relationships between $N_y$ and $N_x$ in a bivariate system.}
		\label{fig:mappings-cases}
	\end{figure}
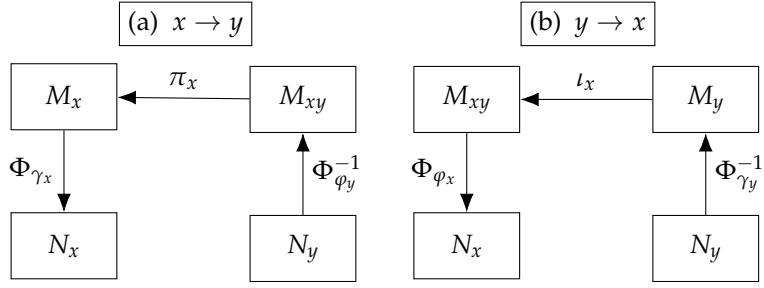
	
	\emph{Case $x \to y$:} When $x \to y$, the forced Takens' theorem guarantees that $\Phi_{\varphi_y} \colon M_{xy} \to N_y$ is diffeomorphic (When $x \to y$, $\Phi_{\varphi_x}$ is not necessarily diffeomorphic. Similarly, when $y \to x$, $\Phi_{\varphi_y}$ is not necessarily diffeomorphic.) We can then construct a mapping from $N_y$ to $N_x$ as $\Psi_{y \to x} = \Phi_{\gamma_x} \circ \pi_x \circ \Phi_{\varphi_y}^{-1}$ (Figure \ref{fig:mappings-cases}(a)). Bounding the stability of this composition using \eqref{eq:stabletakens-Phi-gammax}, \eqref{eq:projection-nonexpansive}, and \eqref{eq:stabletakens-Phiinv-phiy} yields
	\begin{gather*}
		\frac{\norm{\Psi_{y \to x}(p) - \Psi_{y \to x}(q)}_2^2}{\norm{p-q}_2^2} = \frac{\norm{\Phi_{\gamma_x} \circ \pi_x \circ \Phi_{\varphi_y}^{-1}(p) - \Phi_{\gamma_x} \circ \pi_x \circ \Phi_{\varphi_y}^{-1}(q)}_2^2}{\norm{p-q}_2^2} \\
		= \underbrace{\frac{\norm{\Phi_{\gamma_x} \circ \pi_x \circ \Phi_{\varphi_y}^{-1}(p) - \Phi_{\gamma_x} \circ \pi_x \circ \Phi_{\varphi_y}^{-1}(q)}_2^2}{\norm{\pi_x \circ \Phi_{\varphi_y}^{-1}(p) - \pi_x \circ \Phi_{\varphi_y}^{-1}(q)}_2^2}}_{\in~ [l_{\gamma_x},~u_{\gamma_x}]} \times \underbrace{\frac{\norm{\pi_x \circ \Phi_{\varphi_y}^{-1}(p) - \pi_x \circ \Phi_{\varphi_y}^{-1}(q)}_2^2}{\norm{\Phi_{\varphi_y}^{-1}(p) - \Phi_{\varphi_y}^{-1}(q)}_2^2}}_{\in~ [0,~ 1]} \times \underbrace{\frac{\norm{\Phi_{\varphi_y}^{-1}(p) - \Phi_{\varphi_y}^{-1}(q)}_2^2}{\norm{p-q}_2^2}}_{\in~ [u_{\varphi_y}^{-1},~ l_{\varphi_y}^{-1}]},
	\end{gather*}
	so
	\begin{gather*}
		0 = l_{\gamma_x} \times 0 \times u_{\varphi_y}^{-1} \leq \frac{\norm{\Psi_{y \to x}(p) - \Psi_{y \to x}(q)}_2^2}{\norm{p-q}_2^2} \leq u_{\gamma_x} \times 1 \times l_{\varphi_y}^{-1} = \frac{u_{\gamma_x}}{l_{\varphi_y}}
	\end{gather*}
	for all $p,q \in N_y$.
	
	\emph{Case $y \to x$:} When $y \to x$, the forced Takens' theorem guarantees that $\Phi_{\varphi_x} \colon M_{xy} \to N_x$ is diffeomorphic. In this case a mapping from $N_y$ to $N_x$ can be constructed as $\Psi_{y \to x} = \Phi_{\varphi_x} \circ \iota_x \circ \Phi_{\gamma_y}^{-1}$ (Figure \ref{fig:mappings-cases}(b)). Bounding the stability of this composition using \eqref{eq:stabletakens-Phi-phix}, \eqref{eq:inclusion-map-noncontractive}, and \eqref{eq:stabletakens-Phiinv-gammay} yields
	\begin{gather*}
		\frac{\norm{\Psi_{y \to x}(p) - \Psi_{y \to x}(q)}_2^2}{\norm{p-q}_2^2} = \frac{\norm{\Phi_{\varphi_x} \circ \iota_x \circ \Phi_{\gamma_y}^{-1}(p) - \Phi_{\varphi_x} \circ \iota_x \circ \Phi_{\gamma_y}^{-1}(q)}_2^2}{\norm{p-q}_2^2} \\
		= \underbrace{\frac{\norm{\Phi_{\varphi_x} \circ \iota_x \circ \Phi_{\gamma_y}^{-1}(p) - \Phi_{\varphi_x} \circ \iota_x \circ \Phi_{\gamma_y}^{-1}(q)}_2^2}{\norm{\iota_x \circ \Phi_{\gamma_y}^{-1}(p) - \iota_x \circ \Phi_{\gamma_y}^{-1}(q)}_2^2}}_{\in~ [l_{\varphi_x},~u_{\varphi_x}]} \times \underbrace{\frac{\norm{\iota_x \circ \Phi_{\gamma_y}^{-1}(p) - \iota_x \circ \Phi_{\gamma_y}^{-1}(q)}_2^2}{\norm{\Phi_{\gamma_y}^{-1}(p) - \Phi_{\gamma_y}^{-1}(q)}_2^2}}_{\geq~1} \times \underbrace{\frac{\norm{\Phi_{\gamma_y}^{-1}(p) - \Phi_{\gamma_y}^{-1}(q)}_2^2}{\norm{p-q}_2^2}}_{\in~ [u_{\gamma_y}^{-1},~ l_{\gamma_y}^{-1}]},
	\end{gather*}
	so
	\begin{gather*}
		\frac{l_{\varphi_x}}{u_{\gamma_y}} = l_{\varphi_x} \times 1 \times u_{\gamma_y}^{-1} \leq \frac{\norm{\Psi_{y \to x}(p) - \Psi_{y \to x}(q)}_2^2}{\norm{p-q}_2^2} \leq u_{\varphi_x} \times \infty \times l_{\gamma_y}^{-1} = \infty
	\end{gather*}
	for all $p,q \in N_y$.
\end{proof}

\section{Proof of Proposition \ref{thm:causality-stable-takens-iff}}
\label{sup:proofs-causality-stable-takens-iff}

Proposition \ref{thm:causality-stable-takens-iff} provides a test for the existence of causal structure using Assumptions \ref{ass:inclusion-map-expansivity} and \ref{ass:causal-sufficiency}.

\begin{proof}
	($\implies$) Same as case $x \to y$ of Proposition \ref{thm:causality-stable-takens-1}. This direction does not require Assumptions \ref{ass:inclusion-map-expansivity} or \ref{ass:causal-sufficiency}.
	
	($\impliedby$) By assumption we have that for all $p,q \in N_y$,
	\begin{equation*}
		\frac{\norm{\Psi_{y \to x}(p) - \Psi_{y \to x}(q)}_2^2}{\norm{p-q}_2^2} \leq \frac{u_{\gamma_x}}{l_{\varphi_y}},
	\end{equation*}
	so by Assumption \ref{ass:causal-sufficiency} either $x \to y$ or $y \to x$.
	
	Assume for contradiction that $y \to x$ so that the mapping from $N_y$ to $N_x$ is $\Psi_{y \to x} = \Phi_{\varphi_x} \circ \iota_x \circ \Phi_{\gamma_y}^{-1}$ (Figure \ref{fig:mappings-cases}(b)). Using the isometry constants \eqref{eq:stabletakens-Phi-phix} and \eqref{eq:stabletakens-Phiinv-gammay}, we have
	\begin{gather*}
		\frac{\norm{\Psi_{y \to x}(p) - \Psi_{y \to x}(q)}_2^2}{\norm{p-q}_2^2} = \frac{\norm{\Phi_{\varphi_x} \circ \iota_x \circ \Phi_{\gamma_y}^{-1}(p) - \Phi_{\varphi_x} \circ \iota_x \circ \Phi_{\gamma_y}^{-1}(q)}_2^2}{\norm{p-q}_2^2} \\
		= \underbrace{\frac{\norm{\Phi_{\varphi_x} \circ \iota_x \circ \Phi_{\gamma_y}^{-1}(p) - \Phi_{\varphi_x} \circ \iota_x \circ \Phi_{\gamma_y}^{-1}(q)}_2^2}{\norm{\iota_x \circ \Phi_{\gamma_y}^{-1}(p) - \iota_x \circ \Phi_{\gamma_y}^{-1}(q)}_2^2}}_{\geq~ l_{\varphi_x}} \times \frac{\norm{\iota_x \circ \Phi_{\gamma_y}^{-1}(p) - \iota_x \circ \Phi_{\gamma_y}^{-1}(q)}_2^2}{\norm{\Phi_{\gamma_y}^{-1}(p) - \Phi_{\gamma_y}^{-1}(q)}_2^2} \times \underbrace{\frac{\norm{\Phi_{\gamma_y}^{-1}(p) - \Phi_{\gamma_y}^{-1}(q)}_2^2}{\norm{p-q}_2^2}}_{\geq~ u_{\gamma_y}^{-1}} \leq \frac{u_{\gamma_x}}{l_{\varphi_y}}
	\end{gather*}
	so
	\begin{equation*}
		\frac{\norm{\iota_x \circ \Phi_{\gamma_y}^{-1}(p) - \iota_x \circ \Phi_{\gamma_y}^{-1}(q)}_2^2}{\norm{\Phi_{\gamma_y}^{-1}(p) - \Phi_{\gamma_y}^{-1}(q)}_2^2} \leq \frac{u_{\gamma_x} u_{\gamma_y}}{l_{\varphi_x} l_{\varphi_y}}
	\end{equation*}
	for all $p, q \in N_y$. Since $\Phi_{\gamma_y}^{-1}$ is a bijection between $N_y$ and $M_y$ by Takens' theorem, we have
	\begin{equation*}
		\frac{\norm{\iota_x(p) - \iota_x(q)}_2^2}{\norm{p - q}_2^2} \leq \frac{u_{\gamma_x} u_{\gamma_y}}{l_{\varphi_x} l_{\varphi_y}}
	\end{equation*}
	for all $p, q \in M_y$, which provides the contradiction with Assumption \ref{ass:inclusion-map-expansivity} needed to complete the proof.
\end{proof}

\end{document}